\documentclass[journal,twoside,web]{ieeecolor}
\usepackage{generic}
\usepackage{cite}
\usepackage{textcomp}

\usepackage{graphicx}
\usepackage{xspace}
\usepackage{float}

\usepackage[dvipsnames]{xcolor}
\usepackage{tikz}

\usepackage{amsmath,amssymb,amsfonts}
\usepackage{mathrsfs}
\usepackage{url}
\usepackage{arydshln}
\usepackage{empheq}
\usepackage{caption}
\usepackage{subcaption}

\allowdisplaybreaks

\DeclareRobustCommand{\legendline}[1]{\hspace{-3.2pt}
\tikz[#1,line width=2pt,baseline=-0.5ex]{\draw (0,0) -- (.35,0);}
\hspace{-3.2pt}}

\definecolor{mblue}{rgb}{0,0.4470,0.7410}
\definecolor{morange}{rgb}{0.8500,0.3250,0.0980}
\definecolor{myellow}{rgb}{0.9290,0.6940,0.1250}
\definecolor{mpurple}{rgb}{0.4940,0.1840,0.5560}
\definecolor{mgreen}{rgb}{0.4660,0.6740,0.1880}
\definecolor{mcyan}{rgb}{0.3010,0.7450,0.9330}
\definecolor{mred}{rgb}{0.6350,0.0780,0.1840}
\definecolor{mgreenblue}{rgb}{0.0,1.0,0.5}

\newcommand{\col}{\mathrm{col}}
\newcommand{\mr}[1]{\mathrm{#1}}
\newcommand{\m}[1]{\mathcal{#1}}
\newcommand{\mc}[1]{\m{#1}}
\newcommand{\mb}[1]{\mathbb{#1}}

\DeclareFontFamily{OT1}{pzc}{}
\DeclareFontShape{OT1}{pzc}{m}{it}{<-> s * [1.150] pzcmi7t}{}
\DeclareMathAlphabet{\mathpzc}{OT1}{pzc}{m}{it}

\DeclareMathOperator*{\argmin}{arg\,min}

\newcommand{\pnote}{p}

\newcommand\norm[1]{\left\lVert#1\right\rVert}

\newcommand{\ltwo}{\ensuremath{\mathcal{L}_2}\xspace}
\newcommand{\litwo}{\ensuremath{\mathcal{L}_{\mathrm{i}2}}\xspace}

\newcommand{\genState}{x_\mathrm{\pnote}}
\newcommand{\genStateDot}{\dot{x}_\mathrm{\pnote}}
\newcommand{\genStateInitial}{x_\mathrm{\pnote,0}}
\newcommand{\genStateSize}{n_\mathrm{x_p}}
\newcommand{\genStateSet}{\mathbb{X}_\mathrm{\pnote}}

\newcommand{\genInput}{u}
\newcommand{\genInputSize}{n_\mathrm{u}}
\newcommand{\genInputSet}{\mathbb{U}}

\newcommand{\genOutput}{y}
\newcommand{\genOutputSize}{n_\mathrm{y}}
\newcommand{\genOutputSet}{\mathbb{Y}}

\newcommand{\genF}{f_\mathrm{\pnote}}
\newcommand{\genHz}{h_\mathrm{\pnote,z}}
\newcommand{\genHy}{h_\mathrm{\pnote,y}}

\newcommand{\fullstate}{x}  %{\chi}

\newcommand{\sche}{p}
\newcommand{\schedim}{n_\mathrm{\sche}}
\newcommand{\intsche}{\rho}

\newcommand{\st}[1]{#1^*}	% steady-state trajectory notation

\newcommand{\traj}{\xi}
\newcommand{\stTraj}{\st{\traj}}	% full steady-state traj.

\newtheorem{lem}{Lemma}
\newtheorem{defn}{Definition}
\newtheorem{thm}{Theorem}
\newtheorem{rem}{Remark}
\newtheorem{prop}{Proposition}
\newtheorem{cor}{Corollary}
\newtheorem{prob}{Problem Statement}
\newtheorem{asum}{Assumption}

\def\extendedversion{1} % uncomment this line for the extended version

\newcommand{\extver}[2]{%
  \ifx\extendedversion\undefined%
	#2%
  \else%
    #1%
  \fi%
}

\begin{document}
\title{Nonlinear Tracking and Rejection using Linear Parameter-Varying Control}
%\title{Nonlinear Tracking and Rejection using Linear Parameter-Varying Control \extver{\\\vspace{.3em}\normalsize{(Extend Version)}\vspace{-.5em}}{}}
\author{Patrick J. W. Koelewijn, \IEEEmembership{Student Member, IEEE}, Roland T\'oth, \IEEEmembership{Senior Member, IEEE}, Henk Nijmeijer, \IEEEmembership{Fellow, IEEE}, and Siep Weiland
\thanks{This work has received funding from the European Research Council (ERC) under the European Union Horizon 2020 research and innovation programme (grant agreement No 714663) and was also supported by the European Union within the framework of the National Laboratory for Autonomous Systems (RRF-2.3.1-21-2022-00002).}
\thanks{P.J.W. Koelewijn is with the Control Systems Group, Department of Electrical Engineering, Eindhoven University of Technology, P.O. Box 513, 5600 MB Eindhoven, The Netherlands (e-mail: p.j.w.koelewijn@tue.nl).}
\thanks{R. T\'oth is with the Control Systems Group, Department of Electrical Engineering, Eindhoven University of Technology, P.O. Box 513, 5600 MB Eindhoven, The Netherlands and with the Systems and Control  Laboratory, Institute for Computer Science and Control, 1111 Budapest, Hungary (e-mail: r.toth@tue.nl).}
\thanks{H. Nijmeijer is with the Dynamics and Control Group, Department of Mechanical Engineering, Eindhoven University of Technology, P.O. Box 513, 5600 MB Eindhoven, The Netherlands (e-mail: h.nijmeijer@tue.nl).}
\thanks{S. Weiland is with the Control Systems Group, Department of Electrical Engineering, Eindhoven University of Technology, P.O. Box 513, 5600 MB Eindhoven, The Netherlands (e-mail: s.weiland@tue.nl).} \vskip -3mm}

\maketitle
\begin{abstract} 
The \emph{Linear Parameter-Varying} (LPV) framework has been introduced with the intention to provide stability and performance guarantees for analysis and controller synthesis for \emph{Nonlinear} (NL) systems via convex methods. By extending results of the Linear Time-Invariant framework, mainly based on quadratic stability and performance using dissipativity theory, it has been assumed that they generalize tracking and disturbance rejection guarantees for NL systems. However, as has been shown in literature, stability and performance through standard dissipativity is \emph{not} sufficient in order to satisfy the desired guarantees in case of reference tracking and disturbance rejection for nonlinear systems. We propose to solve this problem by the application of incremental dissipativity, which \emph{does} ensure these specifications. A novel approach is proposed to synthesize and realize an NL controller which is able to guarantee incremental stability and performance for NL systems via convex optimization using methods from the LPV framework. Through simulations and experiments, the presented method is compared to standard LPV controller designs, showing significant performance improvements.
\end{abstract}


%%%%% Introduction %%%%%%
\section{Introduction}\label{sec:Introduction}
\IEEEPARstart{T}{he} control of \emph{Nonlinear} (NL) systems has been an intense, ongoing field of research since the early 1970's, and it is still to this date. So far, no systematic way has been found to perform controller synthesis for general NL systems with performance shaping, compared to the class of \emph{Linear Time-Invariant} (LTI) systems where several systematic approaches exist to design or synthesize controllers. The first attempts to transfer the systematic results of the LTI framework to the NL domain was done in the form of heuristic gain-scheduling. In gain-scheduling, the controller changes between a collection of LTI controllers designed at different local operating points of the NL system. In \cite{Shamma1988}, Shamma further developed this into the framework of \emph{Linear Parameter-Varying} (LPV) systems, where the solution set (behavior) of the NL system is captured by the parameter variation of a proxy linear system. Variation of this system, expressed by a measurable variable called the scheduling-variable, is used to describe the original NL behavior and ensure stability and performance guarantees for the NL system. These concepts were then extended to LPV controller synthesis based on results from LTI $H_\infty$-control \cite{Apkarian1995,Packard1994,Scherer2001,Wu1995}, ensuring \ltwo-gain stability and performance conditions. 

The main advantage of using LPV systems to represent NL systems is that the LTI stability and performance concepts, which have been extended to the LPV framework, also can be used for the NL case. By showing that \ltwo-gain stability and performance guarantees do hold for set-point control of NL systems using the LPV framework \cite{Mohammadpour2012} and based on many successful applications of LPV control in practice, it was assumed that these implications naturally hold true for tracking and disturbance rejection specifications. However, as has been shown in \cite{Koelewijn2020}, such guarantees are not valid in the latter case. Namely, in \cite{Koelewijn2020}, it is shown that current LPV stability analysis is only able to guarantee asymptotic stability of a single equilibrium point of the NL system. As a consequence, tracking and rejection based on `standard' LPV control may run into problems, which has also been exemplified in \cite{Scorletti2015}. As a solution, it was proposed that using the notion of incremental \ltwo-gain stability and performance in synthesis, tracking and rejection specifications \emph{can} be ensured.  

A first attempt to use the concept of incremental stability to perform controller synthesis in conjunction with the LPV framework was in \cite{Scorletti2015}. Using the results from \cite{DeHillerin2011}, this work has provided a control synthesis method where the controller itself is restricted to be LTI with an extra input being the scheduling-variable. However, the lack of a multiplicative relationship between the state and the scheduling-variable in the control structure is a heavy limitation compared to standard LPV control. Despite this restriction, the general benefits of the alternative design have been clearly visible from the results. 

Besides the incremental stability concept, similar stability concepts have also developed such as {\emph{contraction}} \cite{Lohmiller1998} and {\emph{convergence}} \cite{Pavlov2006}. These concepts have proven to be not only relevant for tracking and disturbance rejection, but also for many other NL control problems such as synchronization and observer design \cite{Pavlov2006}. In practice, controller design has been accomplished using both contraction and convergence theory, but often relying on complex procedures to perform synthesis. Recently, a convex synthesis framework has been introduced for state feedback design to achieve contraction \cite{Manchester2018,Wang2020}. A key ingredient that these concepts use is that incremental concepts are analyzed or ensured through the use of the so-called differential form of the system, also referred to as variational \cite{Crouch1987,ReyesBaez2019} or differential dynamics \cite{Manchester2018} in literature, which describes the dynamics of the variation along the trajectories of the system. In literature it has been shown how dissipativity properties of the differential form of a system imply \emph{incremental} dissipativity properties of the (original) system \cite{Verhoek2020}.

In this work, our contribution is the development of a systematic output feedback controller synthesis framework to ensure incremental stability and dissipativity based performance for NL systems. A key ingredient to achieve this is the use of the differential form in order to imply incremental stability and dissipativity properties of the (closed-loop) system. We achieve our contribution through the following three key sub-contributions: (i) proposing a methodology and a performance shaping framework to synthesize an output feedback controller for the differential form of the system by exploiting computationally efficient LPV methods, (ii) introducing a realization method for the controller designed for the differential form of the system to get a nonlinear controller that can be implemented for regulating the (original) target system, (iii) rigorous proofs that the obtained controller ensures closed-loop incremental stability and dissipativity based performance specs with the system.

Compared to previous work, we extend the results in \cite{Manchester2018} which use state feedback to ensure \ltwo-gain performance to \emph{output} feedback for \emph{general} quadratic performance. Moreover, we present how the LPV framework can be used effectively to synthesize the output feedback controller in a computationally efficient manner. Compared to \cite{Scorletti2015}, in which the resulting controller is limited to an LTI structure, our proposed controller has full multiplicative relationship between the controller state and scheduling-variable, similar to a standard LPV controller, hence, potentially allowing to achieve better performance. The overall capabilities of the design approach are demonstrated in simulation examples and via experimental studies. 

The paper is structured as follows. First, in Section \ref{sec:preliminaries}, preliminary definitions and theorems are given on standard, incremental, and differential stability and performance notions. In Section \ref{sec:Problem}, a formal definition of the problem statement of this paper is given. Section \ref{sec:IncrFramework} describes the proposed framework used to analyze and synthesize NL controllers ensuring incremental stability and dissipativity based performance via convex optimization. In Section \ref{sec:Examples}, examples are given on the application  of the developed control method. Finally, in Section \ref{sec:Conclusion}, conclusions on the presented results are drawn and future research recommendations are given.
\pagebreak
\subsubsection*{Notation}%
$\mathbb{R}$ is the set of real numbers, while $\mathbb{R}_+\subset\mathbb{R}$ is the set of non-negative reals. $\mathscr{L}_2^q$ is the space of square integrable real-valued functions $\mathbb{R}_+ \rightarrow \mathbb{R}^q$ with norm $\norm{f}_2 = \sqrt{\smash[b]{\int_0^\infty \norm{f(t)}^2 dt}}$, where $\norm{\star}$ is the Euclidean (vector) norm. The notation $x\in \mathbb{X}^{\mathbb{R}_+}$ denotes a signal $x:\mathbb{R}_+ \rightarrow \mathbb{X}$, such that $x(t)\in  \mathbb{X}$, $\forall t \in \mathbb{R}_+$. A function is of class $\mathcal{C}_n$, if its first $n$ derivatives exist and are continuous almost everywhere. For a matrix $A\in\mb{R}^{n\times n}$, the notation $A \succ 0$ $(A \succeq 0)$ indicates that $A$ is symmetric and positive \mbox{(semi-)definite}, while $A \prec 0$ $(A \preceq 0)$ indicates that $A$ is symmetric and negative \mbox{(semi-)definite}. The identity matrix of size $N$ is denoted by $I_N$. We denote the column vector $\left[x_1^\top \, \cdots \, x_n^\top\right]^\top$ by $\col(x_1,\dots,x_n)$, while $\mr{diag}(x_1,,\dots,x_n)$ stands for diagonal concatenation of $x_1,\,\dots,\,x_n$ into a matrix. %, i.e. \scalemath{.8}{\begin{bmatrix}
%	x_1 & &\\&\ddots &\\&&x_n
%\end{bmatrix}}.
%\pagebreak

%%%%%%%%%%%%%%%%%%%%%%%%%%%%%%%%%%%%%%%%%%%%%%%%%%%%%%%%%%%%%%%%%%%%%%%%%%%%%%%%%%%%%%%%%%%%%%%%%
%%%%% Preliminaries %%%%%% 
\section{Preliminaries}\label{sec:preliminaries}
\subsection{Stability and performance of nonlinear systems}
%[changed everything to use $w$ and $z$ in this section]}
Consider a nonlinear dynamic system given by
\begin{equation}\label{eq:nonlinsys}
\Sigma:  \left \lbrace \begin{aligned}
\dot{x}(t) &= f(x(t),w(t));\\
z(t) &= h(x(t),w(t));
\end{aligned} \right.
\end{equation}
where $w\in \mathbb{W} ^{\mathbb{R}_+}$ with $\mathbb{W} \subseteq \mathbb{R}^{n_\mathrm{w}}$ is the input, $x\in\mathcal{C}_{1}^{n_\mathrm{x}}$ is the state variable associated with the considered state-space representation of the system with $x(t) \in \mathbb{X} \subseteq \mathbb{R}^{n_\mathrm{x}}$ and with arbitrary initial condition $x(0) = x_0 \in \mathbb{X}$, and $z\in \mathbb{Z} ^{\mathbb{R}_+}$ with $\mathbb{Z} \subseteq \mathbb{R}^{n_\mathrm{z}}$ is the output.  From a viewpoint of controller synthesis discussed later, $w$ and $z$ can be also seen as general disturbance inputs and performance outputs of the system.  %\todo{[already mention here that $w$ and $z$ are performance channels?]}
%In this paper, we will primarily consider performance specifications in terms of \ltwo-gain bounds,
 It is assumed that solutions $(x,w,z)$ satisfy \eqref{eq:nonlinsys} in the ordinary sense %and are restricted to have \emph{left-compact support} (LCS),
%, with $u \in \mathscr{L}_{2}^{n_\mr{u}}$, 
and $\mathbb{X}$ with $\mathbb{W}$ are considered to be open sets containing the origin. The functions $f:\mathbb{X} \times \mathbb{W} \rightarrow \mb{R}^{n_\mr{x}}$ and $h:\mathbb{X} \times \mathbb{W} \rightarrow \mathbb{Z}$ are assumed to be Lipschitz continuous with $f(0,0)=0$ and $h(0,0)=0$ and to be such that 
% for all input functions $u \in \mathscr{L}_{2}^{n_\mathrm{u}}$ and 
 for all initial conditions $x_0 \in \mathbb{X}$, there is a unique solution $(x,w,z)\in (\mathbb{X}\times\mathbb{W}\times\mathbb{Z})^{\mathbb{R}_{+}}$ which is forward complete. The set of solutions is defined as \vspace{-1mm}
\begin{multline}\label{eq:nonlinsol}
%\begin{aligned}
\mathfrak{B} := \Big\lbrace (x,w,z)\in \big(\mathbb{X}\times \mathbb{W} \times %&
\mathbb{Z}\big)^{\mathbb{R}_{+}} \mathrel{\big|} x\in\mathcal{C}_1^{n_\mr{x}}, \\[-2mm]
%&
\left(x,w,z\right)\,\text{satisfies \eqref{eq:nonlinsys}}\Big\rbrace. 
%\end{aligned}
\end{multline} \vskip -1mm \noindent
%\todo{[Do we require $x\in\mc{C}_1$ or do we just need it to be differentiable almost everywhere (or do we not want to spend to much time on this)]}
Note that this also implicitly restricts the class of input functions that we consider (e.g. being piecewise continuous) in some sense, as they should result in solutions that are in $\mathfrak{B}$.
Furthermore, we define $\mathfrak{B}_\mr{x,w} = \pi_\mr{x,w} \mathfrak{B}$ where $\pi_\mr{x,w}$ denotes the projection $(x,w) = \pi_\mr{x,w}(x,w,z)$.
Introduce also $\Sigma$, the operator representing the dynamic relationship of the system. $\Sigma(w,x_0)\in\mathbb{Z}^{\mathbb{R}_{+}}$ gives the output solution $z$ of \eqref{eq:nonlinsys} for an input $w\in\mathbb{W}^{\mathbb{R}_{+}}$ and initial condition $x_0 \in \mathbb{X}$. The state transition map %associated with
of $\Sigma$ %\eqref{eq:nonlinsys} 
is given as $x(t) = \phi_\mathrm{x}(t,t_0,x_0,w)\in\mathbb{X}$, corresponding to $x$ at time $t$ when the system is driven from $x_0 \in \mathbb{X}$ at time $t_0 \in \mathbb{R}_+$ by the input signal $w \in \mathbb{W}^{\mathbb{R}_+} $. 

%
% $\phi_\mathrm{x}: \mathbb{R}_+\times\mathbb{R}_+\times \mathbb{X}\times \mathbb{W}^{\mathbb{R}_+}\to \mathbb{X}$, 
 
In the LPV framework, stability and performance is commonly analyzed jointly through the theory of dissipativity. Dissipativity of a system is defined as follows:
\begin{defn}[Dissipativity \cite{Willems1972}]\label{def:dissip}
 	System $\Sigma$, given by \eqref{eq:nonlinsys}, is dissipative with respect to a supply function $s:\mathbb{W}\times\mathbb{Z}\rightarrow\mathbb{R}$, if there exists a positive-definite storage function $V:\mathbb{X}\rightarrow \mathbb{R}_+$ with $V(0) = 0$ such that 
 	\begin{equation}
 		V(x(t_1))-V(x(t_0)) \leq \int_{t_0}^{t_1}s(w(t),z(t))\, dt,
 	\end{equation}
 	for all trajectories $(x,w,z)\in\mathfrak{B}$ and for all $t_0$, $t_1 \in \mathbb{R_+}$ with $t_0\leq t_1$.
\end{defn}
\begin{thm}[Stability implied by dissipativity]
	If a system $\Sigma$, given by \eqref{eq:nonlinsys}, is dissipative, according to Definition \ref{def:dissip}, the storage function $V$ is in $\mathcal{C}_1$ and the supply function $s$ satisfies $s(0,z)\leq 0,\,\forall z\in\mathbb{Z}$, then the origin, i.e. $x=0$, is a stable equilibrium point of the system $\Sigma$. Furthermore, $V$ qualifies as a Lyapunov function in a neighborhood of the origin. 
\end{thm}
\begin{proof}
See \cite{VanderSchaft2017,Willems1972}.
\end{proof}
Performance of NL systems is commonly expressed in terms of an \ltwo-gain bound in order to analyze and synthesize controllers for NL systems through the LPV framework. The notion of \ltwo-gain is given as follows:
\begin{defn}[\ltwo-gain \cite{VanderSchaft2017}]
\label{def:l2gain}
A system $\Sigma$, given by \eqref{eq:nonlinsys}, is said to have a finite \ltwo-gain, if, for all $(x,w,z)\in\mathfrak{B}$ with $w \in \mathscr{L}_{2}^{n_\mathrm{w}}$, there is a finite $\gamma \geq 0$
and a function $\zeta(x_0)\geq  0,\ \forall x_0\in\mb{X}$ with $\zeta(0) = 0$  such that
\begin{equation}\label{eq:l2gain}
\norm{z}_{2} \leq \gamma \norm{w}_{2}+\zeta(x_0).
\end{equation}
The induced \ltwo-gain of $\Sigma$, denoted by $\norm{\Sigma}_2$, is the infimum of $\gamma$ such that \eqref{eq:l2gain} still holds.
\end{defn} 

Having both dissipativity and the \ltwo-gain of an NL system defined, the following lemma links the two concepts.
\begin{lem}[\ltwo-gain stability]\label{lem:l2gainDissip}
The \ltwo-gain of an NL system $\Sigma$, defined by \eqref{eq:nonlinsys}, is less than or equal to $\gamma\in\mathbb{R}_+$%as defined by Definition \ref{def:l2gain}
, if  $\Sigma$ is dissipative %as defined in Definition \ref{def:dissip} 
with respect to $s(w,z) = \gamma^2\norm{w}^2-\norm{z}^2$. If such a finite $\gamma$ exists, then we will call $\Sigma$ being \ltwo-gain stable.
\end{lem}
\begin{proof}
	See \cite{VanderSchaft2017}.
\end{proof}
%Note that similar statements can be made for connecting other performance notions such as passivity, the generalized generalized $\mathcal{H}_2$-norm and the $\mathcal{L}_\infty$-gain to dissipativity, see e.g. \cite{Scherer2015}. 

This notion of \ltwo-gain stability together with quadratic supply and quadratic (parameter-dependent) storage functions is often used to analyze NL systems and synthesize controllers for them through the LPV framework with powerful convex optimization based methods \cite{Hoffmann2015}. Other performance notions such as $\mathcal{H}_2$ \cite{Scherer1995} and passivity \cite{Polcz2019} have also been successfully formulated in terms of the dissipativity notion and generalized for the LPV framework. Although, dissipativity-based synthesis and analysis have become popular in the LPV literature, in \cite{Koelewijn2020,Scorletti2015} it has been shown that the connected stability and performance notions are equilibrium dependent, i.e., (i) convergence of state trajectories of the unperturbed system is only implied w.r.t.~the origin, see \cite{Koelewijn2020,Willems1971}, (ii) convergence of the perturbed system response to a non-zero equilibrium point or target trajectory is not guaranteed.   
Consequently, `standard' dissipativity is not the proper notion to use for tracking and rejection problems of NL systems. Hence, the question arises what the proper stability and performance 
notion is for such problems and how can we use it %for analysis and controller design} 
without abandoning the successful analysis and controller synthesis machinery of the LPV framework.

%
% because (i) it does not guarantee convergence of the perturbed system response, and (ii) it is equilibrium point dependent: it implies convergence of state trajectories of the unperturbed system only to the origin, see \cite{Koelewijn2020,Willems1971}.  
%
%
%
%\TR{However, as shown in \cite{Koelewijn2020,Scorletti2015}, using such LPV synthesis approaches for NL systems to achieve reference tracking and disturbance rejection can cause undesired behavior of the closed-loop NL system. This behavior is due to the fact that stability, guaranteed by `standard' dissipativity, is not a proper notion for the expected asymptotic stability and tracking performance of NL systems, because (i) it does not guarantee convergence of the perturbed system response, and (ii) it is equilibrium point dependent: it implies convergence of state trajectories of the unperturbed system only to the origin, see \cite{Koelewijn2020,Willems1971}. Other performance notions such as generalized $\mathcal{H}_2$ and passivity, used in LPV synthesis, suffer from the same phenomenon, as the problem is connected to the used concept of stability and dissipativity. }

%%%%%%%%%%%%%%%%%%%%%%%%%%%%%%%%%%%%%%%%%%%%%%%%%%%%%
\subsection{Equilibrium independent stability and performance}
In the nonlinear literature, several equilibrium independent stability notions have been introduced such as incremental stability \cite{Angeli2002}, contraction theory \cite{Lohmiller1998} and convergence theory \cite{Pavlov2006} to tackle similar issues. These concepts define stability with respect to trajectories of the system instead of with respect to a single equilibrium point. This results in a global concept of stability for NL systems, independent of a particular equilibrium point or trajectory, which is hence especially relevant for tracking and rejection problems. Similar extension have also been made to dissipativity and related performance concepts, resulting in incremental dissipativity \cite{Verhoek2020}, incremental \ltwo-gain and incremental passivity \cite{VanderSchaft2017}, to name a few. %\pagebreak

\begin{defn}[Incremental dissipativity \cite{Verhoek2020}]\label{def:incrdissip}
A system $\Sigma$, given by \eqref{eq:nonlinsys}, is incrementally dissipative with respect to a(n) (incremental) supply function $s_\Delta:\mathbb{W}\times\mathbb{W}\times\mathbb{Z}\times\mathbb{Z}\to\mathbb{R}$, if there exists a(n) (incremental) positive-definite storage function $V_\Delta: \mathbb{X}\times\mathbb{X}\to\mathbb{R}_+$ with $V_\Delta(x,x)=0$ such that  
\begin{equation}\label{eq:IDIE}
\begin{aligned}
V_\Delta\big(x(t_1), \tilde{x}(t_1)\big) &- V_\Delta\big(x(t_0), \tilde{x}(t_0)\big) \\ &\le \int_{t_0}^{t_1}\!\!s_\Delta\big(w(t), \tilde w(t),z(t), \tilde z(t)\big)\,d t,
\end{aligned}
\end{equation}
for any two trajectories $(x,w,z)\in\mathfrak{B}$ and $(\tilde{x},\tilde{w},\tilde{z})\in\mathfrak{B}$ and for all $t_0,t_1\in\mathbb{R}_+$ with $t_0 \leq t_1$.
\end{defn}
\begin{thm}[Diff. incr. dissipativity condition] \label{th:2}
	System $\Sigma$, given by \eqref{eq:nonlinsys}, is incrementally dissipative according to Definition \ref{def:incrdissip}, if there is a storage function $V_\Delta\in\m{C}_1$ such that 
%	\begin{equation}
	\begin{multline}
		\frac{\partial V_\Delta}{\partial x}(x,\tilde{x})f(x,w)+\frac{\partial V_\Delta}{\partial \tilde{x}}(x,\tilde{x})f(\tilde{x},\tilde{w})\\ \leq s_\Delta(w,\tilde{w},h(x,w),h(\tilde{x},\tilde{w})), \label{eq:DDIE}
		\end{multline}
%	\end{equation}
	for all %points 
	$x,\tilde{x}\in\mathbb{X}$ and $w,\tilde{w}\in\mathbb{W}$. %We then call $\Sigma$ to be incrementally dissipative on $\mathbb{X}\times\mathbb{U}$.
\end{thm}
\begin{proof}
	See Appendix \ref{App:proofdiffincr}.
\end{proof}
\begin{thm}[Incremental stability]\label{thm:incrstab}
If system $\Sigma$, given by \eqref{eq:nonlinsys}, is incrementally dissipative according to Definition \ref{def:incrdissip} with storage function $V_\Delta\in\m{C}_1$ and supply function $s_\Delta$  of the form
\begin{equation}\label{eq:quadsupp}
	s_\Delta(w,\tilde{w},z,\tilde{z}) = \begin{bmatrix}
		w-\tilde{w}\\z-\tilde{z}
	\end{bmatrix}^\top \begin{bmatrix}
		Q & S\\S^\top & R
	\end{bmatrix}\begin{bmatrix}
		w-\tilde{w}\\z -\tilde{z}
	\end{bmatrix},
\end{equation} satisfying that 
\begin{equation}\label{eq:suppassum}
	s_\Delta(w,w,z,\tilde{z}) < 0,\quad\forall\,w\in \mathbb{W},\ \ \forall\,z,\tilde{z}\in\mathbb{Z},
\end{equation}
then $\Sigma$ is incrementally asymptotically stable (see \cite[Def.~2.1]{Angeli2002} for the formal definition of this notion of stability).
\end{thm}
\begin{proof}
	See Appendix \ref{App:incrstab}.
\end{proof}
\begin{defn}[Invariance]\label{def:invariance}
For system $\Sigma$, we call $\mathcal{X}\subseteq\mathbb{X}$ to be invariant under a given $\mathcal{W}\subseteq\mathbb{W}$, if  $x(t) = \phi_\mathrm{x}(t,0,x_0,w) \in \mathcal{X}$ for all $t\in\mathbb{R}_+$, $x_0\in\mathcal{X}$ and $w\in\mathcal{W}^{\mathbb{R}_+}$.
\end{defn}
\begin{lem}[Convergence]\label{lem:convergence}
	If system $\Sigma$, given by \eqref{eq:nonlinsys}, is incrementally dissipative on $\mathbb{X}\times\mathbb{W}$ for a supply function of the form \eqref{eq:quadsupp} and \eqref{eq:suppassum} and %assuming that there exist 
	there is a set $\mathcal{X}\subseteq\mathbb{X}$ which is compact and invariant under $\mathcal{W}\subseteq\mathbb{W}$, then there is a unique, so-called, steady-state solution $x^* \in \pi_\mr{x}\mathfrak{B}^\mathcal{X,W}$ for any bounded $w^* \in \pi_\mr{w}\mathfrak{B}^\mathcal{X,W}$, i.e. $(x^*,w^*)\in\mathfrak{B}^{\mathcal{X,W}}_\mr{x,w}$, such that any solution $(x,w^*)\in \mathfrak{B}^\mathcal{X,W}_\mr{x,w}$, converges asymptotically towards $(x^*,w^*)$ as $t\rightarrow \infty$. Here $\mathfrak{B}^\mathcal{X,W}$ 
	denotes the restriction of $\mathfrak{B}$
	%and $\mathfrak{B}^\mathcal{X}_\mr{x,u}$ denote the restrictions of the sets $\mathfrak{B}$ and $\mathfrak{B}_\mr{x,u}$
	such that $x(t)\in\mathcal{X}$ and $w(t)\in\mathcal{W}$ for all $t>0$.
\end{lem}%

\begin{proof} 
See Appendix \ref{App:convergence}.
\end{proof}

As the \ltwo-gain (see Definition \ref{def:l2gain}) is a popular performance metric in the LPV case, it is important to consider its incremental formulation which was first introduced in \cite{Zames1966}. This %incremental \ltwo-gain of a system 
is given by the following definition, adapted from \cite{Verhoek2020}:
\begin{defn}[Incremental \ltwo-gain \cite{Verhoek2020}]\label{def:li2gain}
System $\Sigma$, given by \eqref{eq:nonlinsys}, is said to have a finite incremental \ltwo-gain, denoted as \litwo-gain, if for all $(x,w,z),(\tilde{x},\tilde{w},\tilde{z})\in\mathfrak{B}$ with $w,\tilde{w}\in \mathscr{L}_{2}^{n_\mr{w}}$, there is a finite $\gamma \geq 0$ and a function $\zeta(x_0,\tilde{x}_0)\geq  0,\ \forall x_0,\tilde x_0\in\mb{X}$ with $\zeta(x_0,x_0) = 0$  such that 
\begin{equation}\label{eq:incrml2gain}
\norm{z-\tilde{z}}_{2} \leq \gamma \norm{w-\tilde{w}}_{2}+\zeta(x_0,\tilde{x}_0).
\end{equation}
The induced \litwo-gain of $\Sigma$, denoted by $\norm{\Sigma}_{\mathrm{i}2}$, is the infimum of  $\gamma$ such that \eqref{eq:incrml2gain} holds. 
\end{defn}
Note that the \ltwo-gain and \litwo-gain are the same for LTI systems, see \cite{Koelewijn2019}. 
Similar to Lemma \ref{lem:l2gainDissip}, %also for incremental dissipativity and the \litwo-gain 
the following stability implication can be formulated in the incremental sense.
\begin{lem}[\litwo-gain stability]\label{lem:incrl2gainIncrDissip}
The \litwo-gain of an NL system $\Sigma$, defined by \eqref{eq:nonlinsys}, is less than or equal to
 $ \gamma \in\mathbb{R}_+$, %as defined by Definition \ref{def:li2gain} 
if $\Sigma$ is incrementally dissipative %, as defined in Definition \ref{def:incrdissip}, 
with respect to  $s_\Delta(w,\tilde{w},z,\tilde{z}) = \gamma^2\norm{w-\tilde{w}}^2-\norm{z-\tilde{z}}^2$. If such a finite $\gamma$ exists, then $\Sigma$ is incrementally asymptotically stable and hence we will call $\Sigma$ to be \litwo-gain stable.
\end{lem}
\begin{proof}
See Appendix \ref{App:incrl2gaindissip}.
\end{proof}

Note that, for formulation of incremental notions of passivity, the generalized $\mathcal{H}_2$-norm and the $\mathcal{L}_\infty$-gain, similar relations can be shown, see \cite{Verhoek2020} for the details. However, for the sake of compactness, in this work, we will focus on the \litwo-gain only, although our results do hold under general dissipativity relations.

In order to formulate computable analysis results for incremental dissipativity, we can consider the so called  \emph{differential form}\footnote{To study incremental, contraction and convergence properties of NL systems, similar representations as \eqref{eq:difform} have been developed describing the variation of the system along its trajectories, such as variational dynamics \cite{Crouch1987} or by defining the G\^ateaux derivative of the NL system \cite{Fromion2003}.} of the NL system \eqref{eq:nonlinsys}, while we will refer to the original NL system \eqref{eq:nonlinsys} as the \emph{primal form}. 
Assuming that $f,h\in \m{C}_1$, the differential form of \eqref{eq:nonlinsys} is given by
\begin{equation}\label{eq:difform} \hspace{-0.1mm}
\hspace{-.2em}\delta\Sigma\!:\!\! \left \lbrace \! \begin{aligned}
\!\delta \dot{{x}}(t) &= \m{A}(x(t),\!w(t)\hspace{-.1em})\delta x(t)  + \m{B}(x(t),\!w(t)\hspace{-.1em}) \delta w(t);\\
\!\delta z(t) &= \m{C}(x(t),\!w(t)\hspace{-.1em})\delta x(t) + \m{D}(x(t),\!w(t)\hspace{-.1em}) \delta w(t);
\end{aligned}\right.
\end{equation}
with $\m{A}=\frac{\partial f}{\partial x}$, $\m{B}=\frac{\partial f}{\partial w}$, $\m{C}= \frac{\partial h}{\partial x}$, $\m{D}= \frac{\partial h}{\partial w}$
and $( x, w)\in \mathfrak{B}_\mr{x,w}$. Solutions $\delta x\in \mathcal{C}_1^{n_\mr{x}}$ with $\delta x(t) \in \mathbb{R}^{n_\mr{x}}$ and initial condition $\delta x(0)=\delta x_0\in\mathbb{R}^{n_\mr{x}}$, $\delta w \in(\mathbb{R}^{n_\mr{w}})^{\mathbb{R}_+}$ and $\delta z \in(\mathbb{R}^{n_\mr{z}})^{\mathbb{R}_+}$ of \eqref{eq:difform} are assumed to satisfy \eqref{eq:difform} in the ordinary sense. For a $(x,w)\in\mathfrak{B}_\mr{x,w}$, the solution set of \eqref{eq:difform} is defined as
\begin{multline}\label{eq:nonlinsoldiff}
%\begin{aligned}
\mathfrak{B}_\delta^{(x,w)} \! := \!\Big\lbrace\! (\delta x,\delta w, \delta z)\in \left(\mathbb{R}^{n_\mr{x}}\times \mathbb{R}^{n_\mr{w}} \times\mathbb{R}^{n_\mr{z}}\right)^{\mathbb{R}_{+}} \mid \\[-1.5mm]
 %&\hspace{-.5em}
 \delta x\in\mathcal{C}_1^{n_\mr{x}}, \left(\delta x,\delta w,\delta z\right)\,\text{\small satisfy \eqref{eq:difform} along }( x, w) \!\Big\rbrace.
 %\end{aligned}
\end{multline}
Then $
 	\mathfrak{B}_\delta := \bigcup_{(x,w)\in\mathfrak{B}_\mr{x,w}} \mathfrak{B}_\delta^{(x,w)},
$
gives the complete solution set of \eqref{eq:difform}. To understand how the differential form \eqref{eq:difform} connects to \eqref{eq:nonlinsys} and what the $\left(\delta x,\delta w,\delta z\right)$ trajectories represent w.r.t. solutions of the primal system, consider Fig. \ref{fig:traj}. In this figure, a family of smoothly parameterized state trajectories $\bar{x}$, i.e., a homotopy, with $\bar{x}(\lambda)\in\pi_\mr{x}\mathfrak{B}$ for $\lambda\in[0,1]$, is depicted, which describes a transition from the current state trajectory $x$ at $\lambda =1$, i.e., $\bar{x}(1) = x$, to an other state trajectory $\tilde x$ at $\lambda =0$, i.e. $\bar x(0) = \tilde{x}$, of the solution set $\mathfrak{B}$. The parametrization of $\bar{x}$, i.e. the family of transition trajectories, can for example be chosen such that $\bar{x}(\lambda)$ is the shortest path, called the geodesic, under a given measure (e.g., minimal energy path). Then, $\delta x$ is the tangential variation of $x$ along $\lambda$, i.e., $\delta x(t,\lambda) = \frac{\partial \bar{x}(t,\lambda)}{\partial \lambda}$ with $\delta x(t) =\left.\frac{\partial \bar{x}(t,\lambda)}{\partial \lambda}\right\vert_{\lambda=1}$, while $\dot{x}$ corresponds to variation of $x$ along $t$, i.e., $\dot x(t) = \left.\frac{\partial \bar{x}(t,\lambda)}{\partial t}\right\vert_{\lambda=1}$. Similar definitions hold for the input and output trajectories and variations, see \cite{Verhoek2020} for more details.

\begin{figure}
    \centering
    \includegraphics[scale=.9]{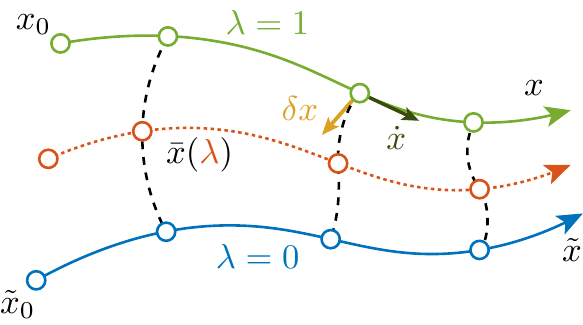}
    \caption{The homotopy based transition path that connects two state trajectories $x$ and $\tilde{x}$ of the solution set $\mathfrak{B}$ together with the corresponding meaning of $\delta x$ and $\dot{x}$.}
    \label{fig:traj} \vspace{-5mm}
\end{figure} 
Based on the differential form of the system, the notion of \emph{differential dissipativity} can also be defined:
\begin{defn}[Differential dissipativity \cite{Verhoek2020}]\label{def:diffdissip}
	System $\Sigma$ with differential form $\delta \Sigma$,  given by \eqref{eq:nonlinsys}  and \eqref{eq:difform}, respectively, is differentially dissipative with respect to the (differential) supply function $s_\delta:\mathbb{R}^{n_\mr{w}}\times\mathbb{R}^{n_\mr{z}}\rightarrow\mathbb{R}$ if there exists a (differential) positive-definite storage function $V_\delta:\mathbb{X}\times\mathbb{R}^{n_\mr{x}}\rightarrow\mathbb{R}_+$ such that 
	\begin{equation}\label{eq:diffDE}
	\begin{aligned}
	V_\delta\big( x(t_1), \delta{x}(t_1)\big) - V_\delta\big(& x(t_0), \delta{x}(t_0)\big) \\ &\le \int_{t_0}^{t_1}\!\!s_\delta\big(\delta w(t),\delta z(t)\big)\,d t,
	\end{aligned}
	\end{equation}
	for all trajectories $(\delta x,\delta w, \delta z)\in \mathfrak{B}_\delta^{(x,w)}$ under every $( x, w)\in \mathfrak{B}_\mr{x,w}$ and for all 
	$t_0,t_1\in\mathbb{R}_+$ with $t_0 \leq t_1$.
\end{defn}

In \cite{Verhoek2020}, it is shown that for  (state-dependent) quadratic storage and quadratic supply functions, differential dissipativity implies incremental dissipativity\footnote{Under the assumption that the quadratic output term of the supply function is negative semi-definite and $\mathbb{X}\times \mathbb{W}$ is convex. 
%the solutions of the NL system take values in a convex set.
}. This allows us to state the following theorem:
\begin{thm}[Diff. induced \litwo-gain stability] \label{thm:diffincrdiss}
	System $\Sigma$, given by \eqref{eq:nonlinsys} with $f,h\in \m{C}_1$ and $\mathbb{X}\times\mathbb{W}\times\mathbb{Z}$ convex, %, and with differential form $\delta \Sigma$, given by \eqref{eq:difform},
	 is \litwo-gain stable (see Lemma \ref{lem:incrl2gainIncrDissip}) with a bounded \litwo-gain $\leq \gamma\in\mathbb{R}_+$, if $\Sigma$ is differentially dissipative as defined in Definition \ref{def:diffdissip} for the supply function $s_\delta(\delta w,\delta z) = \gamma^2\norm{\delta w}^2-\norm{\delta z}^2$ with a quadratic positive-definite storage function $V_\delta( x,\delta x) = \delta x^{\!\top}\! M( x)\delta x$.
\end{thm}
\begin{proof}
	See \cite{Verhoek2020}.
\end{proof}
\begin{rem}
	If the NL system $\Sigma$, given by \eqref{eq:nonlinsys}, is differentially dissipative then this implies dissipativity of its differential form $\delta \Sigma$, given by \eqref{eq:difform}, and vice versa. When $\Sigma$ is differentially dissipative with respect to %a supply function 
	$s_\delta(\delta w, \delta z) = \gamma^2\norm{\delta w}^2-\norm{\delta z}^2$, we will say that the differential form $\delta \Sigma$ has a bounded \ltwo-gain $\gamma$ and it is \ltwo-gain stable. 
\end{rem}

Theorem \ref{thm:diffincrdiss} has useful implications, namely, dissipativity properties of the differential form of an NL system imply incremental dissipativity properties (e.g.,  \litwo-gain boundedness and incremental stability) of the primal form. 
 Note that, for incremental notions of passivity, the generalized $\mathcal{H}_2$-norm and the $\mathcal{L}_\infty$-gain, similar relations and results do hold.

%%%%%%%%%%%%%%%%%%%%%%%%%%%%%%%%%%%%%%%%%%%%%%%%%%%%%%%%%%%%%%%%%%%%%%%%%%%%%%%%%%%%%%%%%%%%%%%%%
%%%%% Problem Statement %%%%%%
\section{Problem statement}\label{sec:Problem}
In this paper we consider the problem of control synthesis for a rather wide class of nonlinear control configurations, described by so-called \emph{generalized plants} $P$ \cite{Apkarian1995}. The objective is to solve the synthesis problem by a novel LPV approach that, via exploiting differential dissipativity, can ensure global stability and performance guarantees for tracking and rejection. As described in Section \ref{sec:Introduction}, current LPV synthesis methods cannot provide such guarantees in general.
A wide range of control structures from feedback and feedforward control to observer design for nonlinear systems can be expressed in the form of the plant
\begin{equation}\label{eq:genplantfull}
    P: \left \lbrace
    \begin{aligned}
    \genStateDot(t) &= f_\mr{\pnote}\left(\genState(t),\genInput(t)\right)+B_\mr{w} w(t);\\
    z(t) &=  \genHz \left(\genState(t),\genInput(t)\right)+D_\mr{zw}w(t);\\
    \genOutput(t) &= \genHy(\genState(t),\genInput(t))+D_\mr{yw}w(t);
    \end{aligned}\right. 
\end{equation}
%\TR{(we could drop here the "p" completely... or keep it, but then the closed loop state can be called x.)}
where $\genState\in\mathcal{C}_1^{\genStateSize}$ is the state with
$\genState(t)\in \mathbb{X}_\mr{\pnote}\subseteq \mathbb{R}^{\genStateSize}$ and initial condition $\genState(0) = \genStateInitial\in\mathbb{R}^{\genStateSize}$, and elements of $w \in \mathbb{W}^{\mathbb{R}_+}$ with $\mathbb{W} \subseteq \mathbb{R}^{n_\mathrm{w}}$  correspond to references, external disturbances, etc., collectively called as \emph{generalized disturbances}, while elements of  $z\in\mathbb{Z}^{\mathbb{R}_+}$ with $\mathbb{Z}\subseteq \mathbb{R}^{n_\mathrm{z}}$  characterize the \emph{generalized performance} (e.g. tracking error, control effort, etc.). Furthermore, we introduce the channels $u$ and $y$ where $\genInput \in\genInputSet^{\mathbb{R}_+}$ with $\genInputSet \subseteq \mathbb{R}^{\genInputSize}$ is the control input and $\genOutput \in \mathbb{Y}_\mr{\pnote}^{\mathbb{R}_+}$ with $\mathbb{Y}_\mr{\pnote}\subseteq \mathbb{R}^{\genOutputSize}$ is the measured output. These represent the channels on which the controller $K$ interacts with $P$.  Additionally, $\genF$, $\genHz$ and  $\genHy$ are assumed to be in $\mathcal{C}_1$. Let us also introduce the solution set of \eqref{eq:genplantfull}, defined as follows
\begin{align}\label{eq:nonlinsolgen}
\mathfrak{B}_\mr{\pnote} := \Big\lbrace\! & (\genState,\genInput,w,z,\genOutput)\in \left(\genStateSet\times \genInputSet\times\mathbb{W}\times\mathbb{Z} \times\genOutputSet\right)^{\mathbb{R}_{+}} \mathrel{\big|}\notag\\[-1mm] 
&\hspace{.5em}x_\mr{\pnote}\in\mathcal{C}_1^{n_\mr{x}},\, 
%\genInput\in\mathscr{L}_{2}^{\genInputSize},w\in\mathscr{L}_{2}^{n_\mr{w}},
\left(\genState,\genInput,w,z,\genOutput\right)\,\text{\small satisfies \eqref{eq:genplantfull}}\Big\rbrace.\hspace{-2em}
\end{align}
Again, like for \eqref{eq:nonlinsys} through $\mathfrak{B}$, this implicitly restricts the class of inputs functions that we consider, as they should be such that the solutions of \eqref{eq:genplantfull} are in $\mathfrak{B}_\mr{\pnote}$. Moreover, introduce $\mathfrak{B}_{\mr{\pnote},\mr{\genState}} = \pi_\mr{\genState}\mathfrak{B}_\mr{\pnote}$. In Fig. \ref{fig:genplant}, an example of such a plant $P$ interconnected with a controller $K$ is given.

The controller $K$ for a given plant (i.e., control configuration) $P$ is considered in the form 
\begin{equation}\label{eq:generalcontroller}
	K: \left \lbrace \begin{aligned}
\dot{x}_\mathrm{k}(t) &= f_\mathrm{k}(x_\mathrm{k}(t),u_\mathrm{k}(t));\\
y_\mathrm{k}(t) &= h_\mathrm{k}(x_\mathrm{k}(t),u_\mathrm{k}(t));
\end{aligned} \right.
\end{equation}
where $x_\mr{k}$ is the state, $u_\mr{k}$ is the input and $y_\mr{k}$ is the output of the controller. The closed-loop interconnection $\mathcal{F}_\mathrm{l}(P,K)$ of $P$ and $K$ through $u_\mr{k} = \genOutput$ and $\genInput = y_\mr{k}$ is an NL system in the form of \eqref{eq:nonlinsys}.
%. Note that $\mathcal{F}_\mathrm{l}(P,K)$ will be a system of the same form as $\Sigma$ \eqref{eq:nonlinsys}.}
\begin{figure}
    \centering
    \includegraphics[scale=.9]{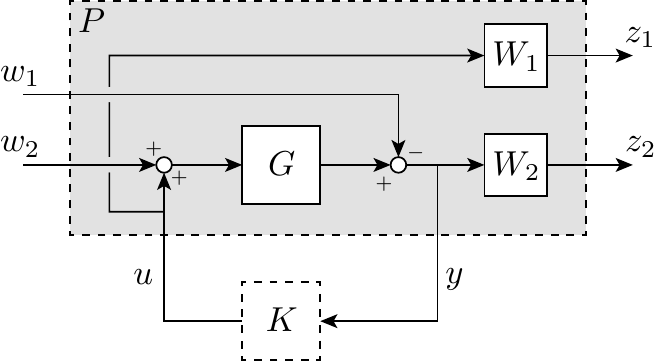}
    \caption{Example of a control configuration in terms of a closed-loop connection of the controller $K$ with a plant $P$, consisting of a nonlinear system $G$ and weighting filters $W_1$ and $W_2$.}%
    \label{fig:genplant} \vspace{-5mm}
\end{figure}

Our objective in this paper is to synthesize $K$ for a given plant $P$, such that the closed-loop interconnection $\mathcal{F}_\mathrm{l}(P,K)$ is incrementally dissipative in terms of Definition \ref{def:incrdissip} under a given supply function $s_\Delta$ that satisfies \eqref{eq:quadsupp} and \eqref{eq:suppassum}, implying closed-loop incremental stability. However, for the sake of compactness of the discussion, we will exemplify the theoretical toolchain only via the incremental \ltwo-gain, although the overall machinery can be easily extended for a general class of performance concepts considered in \cite{Verhoek2020}. These lead to the following problem statement:

\begin{prob} \label{prob:state}
	For a given plant $P$, synthesize $K$ such that the \litwo-gain from $w$ to $z$ of the closed-loop interconnection $\mathcal{F}_\mathrm{l}(P,K)$ is minimized: 
	\begin{equation}
	\begin{gathered}
		\argmin_K\, \gamma \quad \text{s.t.} \\
		\begin{aligned}
\Vert\mathcal{F}_\mathrm{l}( P,  K)(w,\fullstate_0)-\mathcal{F}_{\mathrm{l}}( P, &K)(\tilde{w},\tilde{\fullstate}_0)\Vert_{2} \leq \\&\gamma \norm{w-\tilde{w}}_{2}+\zeta(\fullstate_0,\tilde{\fullstate}_0),
\end{aligned}
	\end{gathered}
	\end{equation}
	for all $\fullstate_0,\tilde \fullstate_0\in\mathbb{X}$ and $w,\tilde{w}\in\mathbb{W}^{\mathbb{R}_+}$ with $w-\tilde{w}\in \mathscr{L}_{2}^{n_\mr{w}}$, where $\fullstate=\col\left(x_\mathrm{\pnote},x_\mr{k}\right)$ %, with $\fullstate(t)\in\mathbb{R}^{n_\fullstate}$, 
	is associated with the state-space representation of the closed-loop $\mathcal{F}_\mathrm{l}(P,K)$ and where $\zeta(\fullstate_0,\tilde \fullstate_0 )\geq  0,\,\forall \fullstate_0,\tilde \fullstate_0\in \mb{X}$ with $\zeta(\fullstate_0,\fullstate_0) = 0$. 
\end{prob}

To ensure that the above given synthesis problem is feasible with a finite $\gamma$, we require $P$ to be a generalized plant:
\begin{defn}[Generalized plant]\label{def:genplant}
	$P$, given by \eqref{eq:genplantfull}, is a generalized plant, if there exists a controller $K$ of the form \eqref{eq:generalcontroller} such that the closed-loop interconnection $\mathcal{F}_\mathrm{l}(P,K)$ is 
%	well-posed and 
	incrementally stable.
\end{defn}
\begin{prop}\label{prop:genplant}
	$P$, given by \eqref{eq:genplantfull}, is a generalized plant, if $(\frac{\partial \genF}{\partial \genState}(\genState,\genInput),\frac{\partial \genF}{\partial \genInput}(\genState,\genInput))$ is stabilizable and $(\frac{\partial \genF}{\partial \genState}(\genState,\genInput),\frac{\partial \genHy}{\partial \genState}(\genState,\genInput))$ is detectable over $\mathbb{X}_\mr{\pnote}\times\mathbb{U}$, see \cite[Section 5.3.2]{Pavlov2006}.
	%	\todo{and are surjective w.r.t.~$\genStateSet$ [@Roland, what was meant again by this?]}, $\mathbb{U}$, while $\genState$ is \todo{locally controllable} via $\genInput$ and \todo{locally observable} via $\genOutput$ for each element in $\genStateSet\times\genInputSet$ (see \cite{Isidori1995,Teel1994}) \todo{[do we need local controllability and observability or controllability and observability of the differential form (I guess they are not the same but not sure if one implies the other)]}. 
\end{prop}

To further simplify our discussion, we will assume that \eqref{eq:genplantfull} can be transformed in terms of the procedure discussed in Appendix \ref{App:A} to the form 
%%%%%%%%%%%%%%%%
\begin{equation}\label{eq:genplant}
    P: \left \lbrace
    \begin{aligned}
    \genStateDot(t) &= \genF\! \left(\genState(t)\right)+B_\mr{w} w(t)+B_\mr{u}\genInput(t);\\
    z(t) &=  \genHz\! \left(\genState(t)\right)+D_\mr{zw}w(t)+D_\mr{zu}\genInput(t);\\
    \genOutput(t) &= C_\mr{y}\genState(t)+D_\mr{yw}w(t);\\
    \end{aligned}\right. 
\end{equation}
for which we will similarly denote its behavior by $\mathfrak{B}_\mr{\pnote}$ and also use $\mathfrak{B}_{\mr{\pnote},\mr{\genState}} = \pi_\mr{\genState}\mathfrak{B}_\mr{\pnote}$.
While \eqref{eq:genplant} may seem restrictive, \eqref{eq:genplantfull} can be always expressed as \eqref{eq:genplant} at cost of increasing the state dimension and requiring the input $u$ to be (piecewise) differentiable~\cite{Nijmeijer2016}. We will see that $P$ in the form of \eqref{eq:genplant} is advantageous to provide a realization of $K$ after synthesis.

%%%%%%%%%%%%%%%%%%%%%%%%%%%%%%%%%%%%%%%%%%%%%%%%%%%%%%%%%%%%%%%%%%%%%%%%%%%%%%%%%%%%%%%%%%%%%%%%%
%%%%% Incremental LPV Framework %%%%%
\pagebreak
\section{Incremental Controller Synthesis}\label{sec:IncrFramework}
\subsection{Main Concept}

To solve Problem \ref{prob:state}, we propose a novel procedure to synthesize an NL controller $K$ that ensures \litwo-gain stability and performance of $\mathcal{F}_\mathrm{l}(P,K)$. The main steps of the method are summarized as follows: 
\begin{enumerate}
    \item   \emph{Differential embedding} step: \label{step1} Given a generalized plant $P$, its differential form $\delta P$ is computed. % The resulting $\delta P$ is embedded in an LPV representation $\delta P_\mathrm{LPV}$. %This step will be referred to as the \emph{differential embedding} step.
    An LPV system $\delta P_\mathrm{LPV}$ is then constructed to represent the resulting $\delta P$, in terms of a so-called LPV embedding\footnote{This will be defined formally later in Definition \ref{def:lpvemb}.}.
    \item \emph{Differential controller synthesis} step: \label{step2} For the LPV embedded differential form $\delta P_\mathrm{LPV}$, %of the generalized plant 
   an LPV controller $\delta K$ is synthesized, ensuring a minimal closed-loop \ltwo-gain. This synthesis is accomplished using standard methods of the LPV framework, e.g., \cite{Apkarian1995,Packard1994,Scherer2001,Wu1995}.    
   %This step will be referred to as the \emph{differential} controller synthesis step;
    \item \emph{Controller realization} step: The synthesized controller $\delta K$ is realized as a primal NL controller $K$ in the form of \eqref{eq:generalcontroller} to be used with the original NL system $P$. %This step will be referred to as the \. 
\end{enumerate}

Our key contributions in the above proposed controller synthesis scheme is the controller realization procedure (Theorem \ref{thm:contrrealize}) and proving that the resulting $K$ solves Problem \ref{prob:state}, i.e., performance and stability guarantees obtained in the differential controller synthesis step do hold in the incremental sense on $\mathcal{F}_\mathrm{l}(P,K)$ (see Theorems \ref{thrm:incric}-\ref{thrm:cl-li2}).
 
%Based on Theorem \ref{thm:diffincrdiss} and our contributions in terms of Theorems \ref{thrm:incric}-\ref{thrm:cl-li2}, it is then proven that \litwo-gain stability and performance for the closed-loop of the original generalized plant and the (to-be-realized) primal form of the synthesized controller is ensured. 
Note that the same procedure can be applied in order to ensure different performance specifications by changing the used performance notion in the differential controller synthesis step, e.g. in order to ensure incremental passivity one would synthesize an LPV controller for the differential form of the generalized plant such that closed-loop passivity is ensured.

\subsection{Separability in the differential domain}

The procedure %above to achieve the objective of minimizing the \litwo-gain of $\mathcal{F}_\mathrm{l}(P,K)$ 
relies on Theorem \ref{thm:diffincrdiss}, which shows that to solve Problem \ref{prob:state} we can equivalently minimize the \ltwo-gain of the differential form of $\mathcal{F}_\mathrm{l}(P,K)$. %In the procedure the plant $P$ and controller $K$ are `transformed' from their primal form to their differential form and vice versa independently. 
Before discussing the steps of the proposed  procedure, we will first show that the differential form of $\mathcal{F}_\mathrm{l}(P,K)$ is equal to $\mathcal{F}_\mr{l}(\delta P,\delta K)$. This significantly simplifies the synthesis procedure, as it allows for independently `transforming' $P$ and $K$ between the primal and %from their primal form to their 
differential domains.
%form. 
\begin{thm}[Closed-loop differential form] \label{thrm:incric} The differen-tial form of the closed-loop system $\mathcal{F}_\mathrm{l}(P,K)$ is equal to the closed-loop interconnection of $\delta P$ and $\delta K$, i.e., $\mathcal{F}_\mathrm{l}(\delta P, \delta K)$, if the interconnection of $P$, $K$, is well-posed i.e., there exists a $\mathcal{C}_1$ function $\breve{h}$, such that $\genInput = h_\mr{k}(x_\mr{k},h_\mr{\pnote,y}(x_\mr{\pnote},\genInput))$ can be expressed as
		$\genInput = \breve{h}(x_\mr{\pnote},x_\mr{k})$.    
\end{thm}
\begin{proof} See Appendix \ref{App:B}. 
\end{proof}

%Next, we will take a look at the first step in the proposed synthesis procedure, i.e. the differential embedding step.

%%%%%%%%%%%%%%%%%%%%%%%%%%%%%%%%%%%%%%%%%%%%%%%%%%%%%
\subsection{Differential embedding}
In the first step of the synthesis procedure, the differential form of the generalized plant $P$ is computed, and the result is embedded in an LPV representation.

Computing the differential form of $P$, given in \eqref{eq:genplant}, results in
\begin{align}\label{eq:genplantincr}
\delta P:\!\left\lbrace\begin{aligned}
    \delta\genStateDot(t) &= \m{A}( x_\mr{\pnote}(t))\delta \genState(t) +{B}_\mr{w}\delta w(t) +{B}_\mr{u} \delta \genInput(t);\\
    \delta z(t) &= \m{C}_\mr{z}( x_\mr{\pnote}(t))\delta \genState(t) +{D}_\mr{zw}\delta w(t)+{D}_\mr{zu} \delta \genInput(t);\\
    \delta \genOutput(t) &= {C}_\mr{y}\delta \genState(t) +{D}_\mr{yw}\delta w(t);\end{aligned}\right. \raisetag{13pt}
\end{align}
where $\m{A}= \frac{\partial \genF}{\partial  x_\mr{\pnote}}$ and $\m{C}_\mr{z} = \frac{\partial \genHz}{\partial  x_\mr{\pnote}}$ with  $x_\mr{\pnote}\in \mathfrak{B}_{\mr{\pnote},\mr{x_p}}$,  %$( x_\mr{\pnote}, \genInput)\in \mathfrak{B}_{\mr{\pnote},\mr{x_p},\mr{u_p}}$,
 $\delta x_\mr{\pnote}\in \mathcal{C}_1^{n_{\mr{x}_\mr{\pnote}}}$ and $\delta x_\mr{\pnote}(t) \in \mathbb{R}^{n_{\mr{x}_\mr{\pnote}}}$ with $\delta \genState(0) = \delta \genStateInitial\in\mathbb{R}^{n_{\mr{x}_\mr{\pnote}}}$, $\delta \genInput(t)\in\mathbb{R}^{\genInputSize}$, $\delta w(t)\in\mathbb{R}^{n_\mr{w}}$, $\delta z(t)\in\mathbb{R}^{n_\mr{z}}$ and $\delta \genOutput(t)\in \mathbb{R}^{\genOutputSize}$. Along a $x_\mr{\pnote}\in \mathfrak{B}_{\mr{\pnote},\mr{x_p}}$ solution of \eqref{eq:genplant}, the set of solutions of \eqref{eq:genplantincr} is %defined as 
\begin{multline} %\label{eq:genplantsoldiff}
\mathfrak{B}_{\delta \mr{\pnote}}^{(x_\mr{\pnote})} := \Big\lbrace (\delta \genState,\delta \genInput,\delta w,\delta z, \delta\genOutput)\in (\mathbb{R}^{\genStateSize}\times \mathbb{R}^{\genInputSize}\\
\times\mathbb{R}^{n_\mr{w}} \times\mathbb{R}^{n_\mr{z}}\times\mathbb{R}^{\genOutputSize})^{\mathbb{R}} \mathrel{\big|} \delta x_\mr{\pnote}\in\mathcal{C}_1^{n_{\mr{x}_\mr{\pnote}}},\, (\delta \genState,\delta \genInput,\\
\delta w,\delta z,\delta \genOutput)\;\text{\small satisfy \eqref{eq:genplantincr} along } \genState\notag\Big\rbrace.
\end{multline}
Then $\mathfrak{B}_\mr{\delta p}\!=\!\!\!\!\!\!\bigcup\limits_{\genState\in \mathfrak{B}_{\mr{\pnote},\mr{\genState}}} \!\!\!\!\!\! \mathfrak{B}_{\delta \mr{\pnote}}^{(\genState)}$ gives the complete solution set of \eqref{eq:genplantincr}.

 Next, we embed \eqref{eq:genplantincr} in an LPV representation: 
\begin{defn}[Differential LPV embedding]
\label{def:lpvemb}
    Given an NL system with primal form \eqref{eq:genplant} and differential form  \eqref{eq:genplantincr}. The LPV state-space representation
\begin{align}\label{eq:genplantincrLPV}
\delta P_\mathrm{LPV}\!:\!\left \lbrace
    \begin{aligned}
    \delta\genStateDot(t) &= {A}(\sche(t))\delta \genState(t) + B_\mr{w}\delta w(t) +{B}_\mr{u}\delta \genInput(t);\\
    \delta z(t) &= {C}_\mr{z}(\sche(t)\!)\delta \genState(t) \!+\!D_\mr{zw}\delta w(t)\!+\! {D}_\mr{zu} \delta \genInput(t);\\
    \delta \genOutput(t) &= {C}_\mr{y} \delta \genState(t) +D_\mr{yw}\delta w(t);
\end{aligned}\right.\raisetag{13pt}
\end{align}
with $A$, $C_\mr{z}$ belonging to a given class of functions $\mathfrak{A}$ (e.g., affine functions) and $\sche(t) \in \mathcal{P}$ being the scheduling variable with a compact and convex  $\mathcal{P} \subset \mathbb{R}^{\schedim}$.  The LPV form \eqref{eq:genplantincrLPV} is called an embedding of \eqref{eq:genplantincr} on the compact region $\mathcal{X}_\mathrm{\pnote}\subset \genStateSet$, if 
there is a function $\psi: \mathbb{R}^{\genStateSize} \rightarrow \mathbb{R}^{\schedim}$ with $\psi \in \mathcal{C}_\mr{1}$ and $\psi(\mathcal{X}_\mathrm{\pnote}) \subseteq \mathcal{P}$, such that  
$A\circ\psi = \m{A}$ (i.e. $A(\psi(x_\mr{\pnote})) = \m{A}(x_\mr{\pnote})$) and $C_\mr{z}\circ\psi = \m{C}_\mr{z}$. 
For a given $\sche \in \m{P}^\mathbb{R_+}$, the set of solutions of \eqref{eq:genplantincrLPV} is given as
\begin{align}\label{eq:genplantsoldiff}
&\mathfrak{B}_\mr{LPV}^{(\sche)} := \Big\lbrace (\delta \genState,\delta \genInput,\delta w,\delta z,\delta \genOutput)\in (\mathbb{R}^{\genStateSize}\times \mathbb{R}^{\genInputSize}\times\mathbb{R}^{n_\mr{w}} \notag\\
&\hspace{3em}\times\mathbb{R}^{n_\mr{z}}\times\mathbb{R}^{\genOutputSize})^{\mathbb{R}_+} \mathrel{\big|} \delta \genState\in\mathcal{C}_1^{n_{\mr{x}_\mr{\pnote}}},\,(\delta \genState,\delta \genInput,\\
&\hspace{7.5em}\delta w,\delta z,\delta \genOutput)\,
\text{\small satisfy \eqref{eq:genplantincrLPV} along } \sche \notag\Big\rbrace.
\end{align}
Then, $\mathfrak{B}_\mr{LPV}:=\bigcup\limits_{\sche\in \m{P}^\mathbb{R}} \mathfrak{B}_\mr{LPV}^{(\sche)}$ gives the complete solution set of \eqref{eq:genplantincrLPV}. For $\mathfrak{B}_{\mr{\pnote},\mr{x_p}}$, define the restriction of the trajectories to $\m{X}_\mr{\pnote}$ as $\mathfrak{B}_{\mr{\pnote},\mr{x_p}}^\m{X}:= \mathfrak{B}_{\mr{\pnote},\mr{x_p}}\cap \m{X}_\mr{\pnote}^\mathbb{R}$. As $\mathfrak{B}^{(\psi(x_\mr{\pnote}))}_\mr{LPV} = \mathfrak{B}^{(x_\mr{\pnote})}_\mr{\delta p}$, we can state $\mathfrak{B}_\mr{\delta p}^\m{X}:=\bigcup\limits_{x_\mr{\pnote}\in \mathfrak{B}_{\mr{\pnote},\mr{x_p}}^\m{X}} \mathfrak{B}_{\delta \mr{\pnote}}^{(x_\mr{\pnote})} \subseteq \bigcup\limits_{\sche\in \m{P}^{\mathbb{R}_+}} \mathfrak{B}_\mr{LPV}^{(\sche)} = \mathfrak{B}_\mr{LPV}$.
\end{defn}
\begin{rem}
For tractable controller synthesis later in Section \ref{sec:diffsynth}, the function $\psi$ is must be chosen such that the resulting dependence of $A$ and $C_\mathrm{z}$ on $\sche$, i.e., the class $\mathfrak{A}$, is either affine, polynomial or rational and $n_\mr{\sche}$ is minimal. Furthermore, $\mathcal{P}$ needs to be chosen such that the LPV representation \eqref{eq:genplantincrLPV} is stabilizable from $\delta\genInput$ and detectable from $\delta\genOutput$ over $\mathcal{P}$. Moreover, $\mathcal{P}$ is also chosen such that it is the smallest convex set in a given complexity class ($n$-vertex polytope, hyper-ellipsoid, etc.) such that $\psi(\mathcal{X}_\mr{\pnote}) \subseteq \mathcal{P}$, in order to minimize the conservativeness of the LPV representation in describing the differential form. See \cite{Kwiatkowski2008,Sadeghzadeh2020,Hoffmann2016} for approaches to fulfill these properties.
\end{rem}

Using the LPV embedding principle for $\delta P$ on the region $\mathcal{X}_\mr{\pnote}$ results in the LPV form of $\delta P$ as in \eqref{eq:genplantincrLPV} where $\sche \in \mathcal{P}^\mathbb{R}$ is assumed to be measurable. Note that in terms of Definition \ref{def:lpvemb}, there exist a function $\psi$ such that $\sche(t) =\psi(x_\mr{\pnote}(t))$ under $x_\mr{\pnote}(t) \in \mathcal{X}_\mr{\pnote} $, with $\psi(\mathcal{X}_\mr{\pnote})\subseteq\mathcal{P}$.

%%%%%%%%%%%%%%%%%%%%%%%%%%%%%%%%%%%%%%%%%%%%%%%%%%%%%
\subsection{Differential synthesis}\label{sec:diffsynth}
As aforementioned, we want to synthesize a controller $K$ in order to minimize the \litwo-gain of $\mathcal{F}_\mathrm{l}(P,K)$. This is is done by first synthesizing a differential controller $\delta K$ such that the \ltwo-gain of $\mathcal{F}_\mathrm{l}(\delta P,\delta K)$ is minimized. Then, later in Section \ref{sec:contrrealiz}, a primal form $K$ of the controller $\delta K$ is realized that preserves the achieved closed-loop properties of $\mathcal{F}_\mathrm{l}(\delta P,\delta K)$. In order to perform controller synthesis for the differential form $\delta P$ %\eqref{eq:genplantincr}, 
the LPV framework is used. More concretely, we synthesize a controller for the LPV embedding of the differential form $\delta P_\mr{LPV}$, given in \eqref{eq:genplantincrLPV}, which was constructed in the differential embedding step in the previous subsection. To achieve this, we can apply our standard \ltwo-gain LPV synthesis techniques on \eqref{eq:genplantincrLPV} such as polytopic or LFT-based LPV synthesis methods (see \cite{Hoffmann2015} for an overview), to synthesize a controller $\delta K$ and ensure \ltwo-gain stability of the closed loop interconnection $\mathcal{F}_\mathrm{l}(\delta P_\mr{LPV},\delta K)$, for all $\sche \in \mathcal{P}^{\mathbb{R}_+}$.
This synthesized controller is assumed to be of the following form
\begin{equation}\label{eq:incrContr}
\delta K: \left \lbrace
    \begin{alignedat}{2}
    \delta \dot{x}_\mathrm{k}(t) &= A_\mathrm{k}(\sche(t)) \delta x_\mathrm{k}(t) &&+ B_\mathrm{k}(\sche(t)) \delta u_\mathrm{k}(t);\\
    \delta y_\mathrm{k}(t) &= C_\mathrm{k}(\sche(t)) \delta x_\mathrm{k}(t) &&+ D_\mathrm{k}(\sche(t)) \delta u_\mathrm{k}(t);
    \end{alignedat}\right. 
\end{equation}
which we will refer to as the differential controller, where $\delta x_\mathrm{k}(t) \in \mathbb{R}^{n_\mathrm{x_k}}$ is the state, $\delta u_\mathrm{k}(t) \in \mathbb{R}^{n_{\mathrm{u_k}}}$ is the input, and $\delta y_\mathrm{k}(t) \in \mathbb{R}^{n_{\mathrm{y_k}}}$ is the output of the controller, respectively and $A_\mathrm{k},\ldots,D_\mathrm{k}\in\mathfrak{A}$ are matrix functions with appropriate dimensions. 

%\TR{Should we include here a invariance check? }

\begin{thm}[Differential closed-loop \ltwo-gain]\label{thrm:diffICL2}
	If control-ler $\delta K$ of the form \eqref{eq:incrContr} ensures bounded \ltwo-gain $\gamma$ of the closed-loop interconnection $\mathcal{F}_\mr{l}(\delta P_\mr{LPV},\delta K)$ for all $\sche\in\mathcal{P}^\mathbb{R}$, then %the closed-loop interconnection 
	$\mathcal{F}_\mr{l}(\delta P,\delta K)$ with $p=\psi(x_\mr{\pnote})$ is also  \ltwo-gain stable with an \ltwo-gain %less than or equal to 
	$\leq\gamma$ for all $x_\mr{\pnote}\in\mathfrak{B}_{\mr{\pnote},\mr{x_p}}^\m{X}$.
\end{thm}
\begin{proof}
See Appendix \ref{App:CB}
\end{proof}
%
%\todo{Add remark/note here that we assume there exists a quadratic storage function (not parameter-varying), which is required for realization? (We only say this in Theorem 7).}
\begin{asum}\label{asum:storassum} %[Changed to assumption]}
	We assume that the controller synthesis has been solved such that $\mathcal{F}_\mr{l}(\delta P,\delta K)$ is dissipative with a quadratic (differential) storage function of the form $V_\delta(\delta \fullstate,\fullstate) = \delta \fullstate^\top M \delta \fullstate$, where $M\succ 0$, i.e., a quadratic $V_\delta$ which is independent of $\fullstate$. This is required for the proposed controller realization procedure in Section \ref{sec:contrrealiz}.
\end{asum}
\begin{rem}\label{rem:weightfilt}
By applying shaping filters on $P$ that consequently appear in $\delta P$, we can shape the closed-loop performance of $\mathcal{F}_\mr{l}( P, K)$, see Fig.~\ref{fig:genplant} and Fig.~\ref{fig:shapefigprim}. If the weighting filters included in $P$ are LTI, then as depicted in Fig.~\ref{fig:shapefig}, the input-output behavior of $W_\mr{w}$ and $W_\mr{z}$ is equivalent to that of $\delta W_\mr{w}$ and $\delta W_\mr{z}$, as the dynamics of the differential form of an LTI system are equivalent to the dynamics of its primal form. This results in a one to one correspondence between the performance shaping of the primal form $\mathcal{F}_\mr{l}( P, K)$ (see Fig.~\ref{fig:shapefigprim}) and performance shaping of the differential form $\mathcal{F}_\mr{l}( \delta P, \delta K)$ (see Fig.~\ref{fig:shapefigdiff}). This significantly simplifies the controller design, as shaping can be directly performed through the differential form $\delta P$ and hence also through its LPV embedding $\delta P_\mr{LPV}$.
\end{rem}

\begin{figure}
	\centering
	\begin{subfigure}[b]{\columnwidth}
		\centering
		\includegraphics[scale=.9]{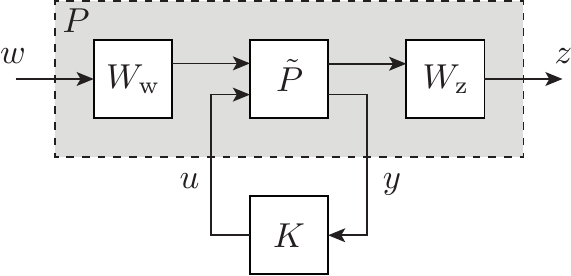}
		\caption{Primal form.}
		\label{fig:shapefigprim}
	\end{subfigure}
	\\[1em]
	\begin{subfigure}[b]{\columnwidth}
		\centering
		\includegraphics[scale=.9]{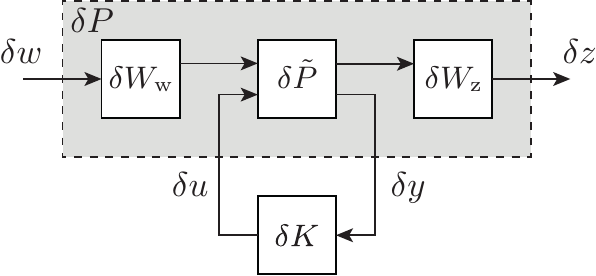}
		\caption{Differential form.}
		\label{fig:shapefigdiff}
	\end{subfigure}
	\caption{Shaping the closed-loop behavior of the primal and the differential form by the use of weighting filters $W_\mr{w}$ and $W_\mr{z}$.}
	\label{fig:shapefig} \vspace{-4mm}
\end{figure}

%%%%%%%%%%%%%%%%%%%%%%%%%%%%%%%%%%%%%%%%%%%%%%%%%%%%%
\subsection{Controller realization}\label{sec:contrrealiz}
We will now describe how to realize the primal form $K$ of the controller for the NL system such that the differential form of $K$ is given by $\delta K$ in \eqref{eq:incrContr} and incremental dissipativity of the closed-loop is ensured. Similar to the approach in \cite{Manchester2018}, we take a path integral based realization, whereby we integrate over the variation, $\lambda$ in Fig. \ref{fig:traj}, in order to converge from the current trajectory towards a known desired (feasible) steady-state trajectory. Namely, to guarantee incremental stability and performance, we consider $\stTraj \triangleq (\st{x_\mr{\pnote}},\st{\genInput},\st{w},\st{z},\st{\genOutput})\in\mathfrak{B}_\mr{\pnote}$ of $P$ to be a \emph{known} trajectory, towards which we want to converge. Let us denote by $\delta\fullstate(t)\in\mathbb{R}^{n_\mathrm{\fullstate}}$ the state associated with $\mathcal{F}_\mathrm{l}(\delta P,\delta K)$, analogous to the state $\fullstate$ of the primal form of the closed-loop interconnection $\mathcal{F}_\mathrm{l}(P,K)$.  %\todo{[Add clarifying sentence to avoid confusion, like with previous reviewer?]}
\begin{thm}[Controller realization]\label{thm:contrrealize}
Given a differential controller $\delta K$ in the form of \eqref{eq:incrContr} that ensures closed-loop \ltwo-gain stability of $\mathcal{F}_\mr{l}(\delta P,\delta K)$ under Assumption \ref{asum:storassum}. Let $(\st{x_\mr{\pnote}},\st{\genInput},\st{\genOutput})=\pi_\mr{x_\mr{\pnote},\genInput,\genOutput}\stTraj\in\pi_\mr{x_\mr{\pnote},\genInput,\genOutput}\mathfrak{B}_\mr{\pnote}$ be the (desired) steady-state trajectory of $P$ and consider the nonlinear controller, omitting dependence on time for brevity, 
\begin{equation}\label{eq:Contr}
 K:\! \left \lbrace
    \begin{aligned}
     \Delta\dot{x}_\mathrm{k} &= \m{A}_\mathrm{k}(x_\mr{\pnote}, \st{x_\mr{\pnote}}) \Delta x_\mathrm{k} + \m{B}_\mathrm{k}(x_\mr{\pnote}, \st{x_\mr{\pnote}})  \Delta u_\mathrm{k};\\
     y_\mathrm{k} &= \st{y_\mr{k}}+\m{C}_\mathrm{k}(x_\mr{\pnote}, \st{x_\mr{\pnote}}) \Delta x_\mathrm{k} + \m{D}_\mathrm{k}(x_\mr{\pnote}, \st{x_\mr{\pnote}}) \Delta u_\mathrm{k};
    \end{aligned}\right. 
\end{equation}
with $(\st{y_\mr{k}},\st{u_\mr{k}})\!=\!(\st{\genInput},\st{\genOutput})$, $\Delta x_\mr{k}(t)\in\mathbb{R}^{n_\mathrm{x}}$, $\Delta u_\mr{k}: =u_\mr{k}-\st{u_\mr{k}}$, and
\begin{equation}
\begin{aligned}\label{eq:incrMat}
	\m{A}_\mr{k}(x_\mr{\pnote}, \st{x_\mr{\pnote}})&= \int_0^1\!\!A_\mr{k}\Big(\psi\big(\st{x_\mr{\pnote}}+\lambda(x_\mr{\pnote}-\st{x_\mr{\pnote}})\big)\Big)\,d\lambda,\\ 
	\m{B}_\mr{k}(x_\mr{\pnote}, \st{x_\mr{\pnote}})&= \int_0^1\!\!B_\mr{k}\Big(\psi\big(\st{x_\mr{\pnote}}+\lambda(x_\mr{\pnote}-\st{x_\mr{\pnote}})\big)\Big)\,d\lambda,\\
	\m{C}_\mr{k}(x_\mr{\pnote}, \st{x_\mr{\pnote}})&=\int_0^1\!\!C_\mr{k}\Big(\psi\big(\st{x_\mr{\pnote}}+\lambda(x_\mr{\pnote}-\st{x_\mr{\pnote}})\big)\Big)\,d\lambda,\\
	\m{D}_\mr{k}(x_\mr{\pnote}, \st{x_\mr{\pnote}})&= \int_0^1\!\!D_\mr{k}\Big(\psi\big(\st{x_\mr{\pnote}}+\lambda(x_\mr{\pnote}-\st{x_\mr{\pnote}})\big)\Big)\,d\lambda.
\end{aligned}
\end{equation}
The controller $K$ in \eqref{eq:Contr} is the primal form of $\delta K$ \eqref{eq:incrContr} and the differential form of $K$ is $\delta K$. Hence, $K$ is called the primal realization of $\delta K$.
\end{thm}
\begin{proof} See Appendix \ref{App:C}. 
	\end{proof}
We will refer to the proposed controller $K$ as the \emph{incremental controller}. Even if $\delta K$ is an LPV controller,  realization by Theorem \ref{thm:contrrealize} results in an NL controller through the resubstitution and integration over the scheduling map $\psi$.  %\todo{Furthermore, $K$ is time invariant as $\m{A}_\mr{k},\dots,\m{D}_\mr{k}$ is expressed to be dependent on $t$ only for notational convenience, while it is clear that these matrix functions depend solely on $x_\mr{\pnote}(t)$ and $ \st{x_\mr{\pnote}}(t)$.}
%  \todo{[why ``incremental NL controller'' instead of ``incremental LPV controller''? If you want to remove ``LPV'' I would be more happy with ``incremental controller'' (I will also need to change this everywhere else in the text, highlighted in red)]}

Note that the resulting controller consists of both a direct feedforward action $\st{y_\mr{k}}=\st{\genInput}$, corresponding to the desired steady-state trajectory $\stTraj$, and a feedback action on the measured deviation from the desired steady-state output $\st{\genOutput}$. Hence, for implementation, we require knowledge of $\stTraj$ and hence $\st{w}$, either through direct measurement or estimation. This will be discussed in more detail in Section \ref{sec:distobsv}. Analytical solution of the integrals for $\m{A}_\mr{k},\dots,\m{D}_\mr{k}$ in \eqref{eq:incrMat} might be difficult in some cases, however they can reliably calculated (online) through numerical computation \cite{Atkinson1989}. Moreover, if the scheduling-dependency of $\delta K$ is affine, \eqref{eq:incrMat} can be further simplified.
\begin{cor}\label{cor:affineint}
	Assume that $A_\mr{k},\, \dots,\, D_\mr{k}$ in \eqref{eq:incrContr}, characterizing $\delta K$, are affine in $\sche$, i.e., 
		$A_\mr{k}(\sche) = A_{\mr{k},0}+\sum_{i=1}^{n_\mr{\sche}} A_{\mr{k},i}\,\sche_i.$
	Then, the matrix functions in \eqref{eq:Contr} are given as
	\begin{equation}
	\begin{aligned}
		\m{A}_\mr{k}(x_\mr{\pnote}, \st{x_\mr{\pnote}}) &= \int_0^1 A_\mathrm{k}\Big(\psi\big(\st{x_\mr{\pnote}}+\lambda(x_\mr{\pnote}-\st{x_\mr{\pnote}})\big)\Big)\,d\lambda,\\
		&=A_{\mr{k},0}+\sum_{i=1}^{n_\mr{\sche}} A_{\mr{k},i}\,\intsche_i=A_\mr{k}(\intsche),
	\end{aligned}
	\end{equation}
	where $\intsche = \int_0^1\psi(\st{x_\mr{\pnote}}+\lambda(x_\mr{\pnote}-\st{x_\mr{\pnote}}))\,d\lambda$, and similarly $\m{B}_\mr{k}(x_\mr{\pnote}, \st{x_\mr{\pnote}}) = B_\mr{k}(\intsche)$, $\m{C}_\mr{k}(x_\mr{\pnote}, \st{x_\mr{\pnote}}) = C_\mr{k}(\intsche)$ and $\m{D}_\mr{k}(x_\mr{\pnote}, \st{x_\mr{\pnote}}) = D_\mr{k}(\intsche)$.
\end{cor}

%\TR{Motivate realizability, practical applicability...}

%%%%%%%%%%%%%%%%%%%%%%%%%%%%%%%%%%%%%%%%%%%%%%%%%%%%%
\subsection{Closed-loop incremental stability and performance}\label{sec:incrl2gain}
%\TR{There is no bridging text here, what you want to show here...}
Next, we will show that the proposed controller $K$ ensures closed-loop \litwo-gain stability of $\mathcal{F}_\mr{l}( P, K)$.
%Let us denote by 
%$\pi_{\mr{x}_\mr{\pnote}}$ the projection $x_\mr{\pnote}=\pi_{\mr{x}_\mr{\pnote}}\fullstate$ and by 
%$\phi_\mrfullstate$ the state transition map $\fullstate(t) = \phi_\mrfullstate(t,0,\fullstate_0,w)$ of the closed loop. 
%Moreover, denote by $\phi_\mr{x_p}$ the state transition map $\pi_\mr{x_p}\phi_\fullstate$ such that $x_\mr{\pnote}(t)=\phi_\mr{x_p}(t,0,\fullstate_0,w)$.
\begin{thm}[Closed-loop \litwo-gain stability] \label{thrm:cl-li2}
    Let $\delta K$, given in \eqref{eq:incrContr}, be an LPV controller, synthesized for $\delta P$ in \eqref{eq:genplantincr} of an NL system \eqref{eq:genplant}, which ensures \ltwo-gain stability of the closed-loop $\mathcal{F}_\mathrm{l}(\delta P, \delta K)$ with a bounded \ltwo-gain of $\gamma$ on $\mathcal{X}_\mr{\pnote}$ under a (differential) storage function of the form $V(\fullstate,\delta\fullstate)=\delta\fullstate^\top M\delta\fullstate$ with $M\succ 0$. 
    Consider the set $\mathcal{W}\subseteq \mathbb{W}$, such that there is an open and bounded $\mathcal{X}_\mr{k} \subseteq \mathbb{R}^{n_\mr{x_k}}$ for which $\mathcal{X}=\mathcal{X}_\mr{\pnote} \times \mathcal{X}_\mr{k} $ is invariant, in the sense of Definition \ref{def:invariance}.
%    \begin{multline}
%\fullstate(t)=\phi_\mathrm{x}(t,0,\fullstate_0,w)\in \label{eq:wset}\mathcal{X},\\ \forall \fullstate_0\in\mathcal{X},\ \forall w\in\mathcal{W}^{\mathbb{R}_+},\ \forall t\in\mathbb{R}_+.
%\end{multline}
    Then the controller $K$, given by \eqref{eq:Contr}, ensures \litwo-gain stability of the closed-loop $\mathcal{F}_\mathrm{l}(P, K)$ with a bounded \litwo-gain of $\gamma$ on $\mathcal{X}$ in the following sense. There exists a function $\zeta(\fullstate,\tilde{\fullstate})\geq  0$ with $\zeta(0,0) = 0$ such that 
\begin{equation}\label{eq:closedincrgain}
\begin{aligned}
\Vert\mathcal{F}_\mathrm{l}( P,  K)(w,\fullstate_0)-&\mathcal{F}_{\mathrm{l}}( P,  K)(\tilde{w},\tilde{\fullstate}_0)\Vert_{2} \\ &\leq \gamma \norm{w-\tilde{w}}_{2}+\zeta(\fullstate_0,\tilde{\fullstate}_0),
\end{aligned}
\end{equation}
for all $w,\tilde{w} \in \mathcal{W}^{\mathbb{R}_+}$, with $w-\tilde{w}\in \mathscr{L}_{2}^{n_\mathrm{w}}$ and $\fullstate_0,\tilde{\fullstate}_0 \in \mathcal{X}$. 
%% ORIGINAL %%
    %  By assuming that there exists a set $\mathfrak{W}$, with $w,\st{w}\in\mathfrak{W}$, where 
%\begin{equation}
%\begin{aligned}
%	\mathfrak{W} &:= \lbrace w \in \mathscr{L}_{2}^{n_\mathrm{w}}\mid w(t)\in\mathcal{W},\\
%	&x_\mr{\pnote}(t)=\phi_\mathrm{x}(t,0,\fullstate_0,w)\in \label{eq:wset}\mathcal{X_\mr{\pnote}},\forall \fullstate_0\in\mathfrak{X},\,\forall t\geq0\rbrace, 
%\end{aligned}
%\end{equation}
%then, $K$ given by \eqref{eq:Contr} ensures \litwo-gain stability of the closed-loop $\mathcal{F}_\mathrm{l}(P, K)$ with a bounded \litwo-gain of $\gamma$ on $\mathcal{X}_\mr{\pnote}$ in the following sense: There exists a compact ball $\mathfrak{X}\subset \mathbb{X}_\mathrm{\pnote}\times\mathbb{R}^{n_{\mathrm{x}_\mathrm{k}}}$, and a function $\zeta(\fullstate,\tilde{\fullstate})\geq  0$ with $\zeta(0,0) = 0$ such that 
%\begin{equation}\label{eq:closedincrgain}
%\begin{aligned}
%\Vert\mathcal{F}_\mathrm{l}( P,  K)(w,\fullstate_0)-&\mathcal{F}_{\mathrm{l}}( P,  K)(\st{w},\st{\fullstate}_0)\Vert_{2} \\ &\leq \gamma \norm{w-\st{w}}_{2}+\zeta(\fullstate_0,\st{\fullstate_0}),
%\end{aligned}
%\end{equation}
%for all $w,\st{w} \in \mathfrak{W}$ and $\fullstate_0,\st{\fullstate_0} \in \mathfrak{X}$.   
% END OF ORIGINAL %
\end{thm}
\begin{proof} See Appendix \ref{App:D}.
\end{proof}
\begin{rem}
	The value set of the generalized disturbance signals $\mc{W}$ considered in Theorem \ref{thrm:cl-li2}, i.e., $w(t)\in\mc{W}$, ensures that only generalized disturbances are considered such that $\genState\in\mc{X}_\mr{\pnote}$, which corresponds to the set for which \ltwo-gain stability was verified of the closed-loop differential form, see Section \ref{sec:diffsynth}. Computing this set is a difficult problem which is related to reachability analysis or invariant set computation, however there are numerical tools that can be employed for this purpose, see e.g. \cite{Althoff2013,Maidens2015}. %However, further details are outside of the scope of the paper.
\end{rem}

%%%%%%%%%%%%%%%%%%%%%%%%%%%%%%%%%%%%%%%%%%%%%%%%%%%%%
\subsection{Steady-state solution}\label{sec:distobsv}
\subsubsection{Estimating the steady-state solution}
%\TR{We need a section on the steady state solution computation where the disturbance observer is just a part. I would also draw a parallel w.r.t. trajectory planning in robotics.} 

The realized controller $K$ in terms of \eqref{eq:Contr} consists of a feedforward and a feedback part. The feedforward part $\st{\genInput}=\st{y_\mr{k}}$ corresponds to the steady-state trajectory $\stTraj = (\st{x_\mr{\pnote}},\st{\genInput},\st{w},\st{z},\st{\genOutput})\in\mathfrak{B}_\mr{\pnote}$.  This trajectory is chosen a-priori by the user, based on the desired reference the system needs to follow, similar to trajectory planning in robotics. Computation of $\stTraj$  can be accomplished by using trajectory planning algorithms \cite{Gasparetto2015} or in some cases by analytic solution of the system equations. 

Note that the generalized disturbances  $w$ also influence the steady-state trajectory. The generalized disturbance consists of measurable/known disturbances $w_\mr{m}$, such as references, and unmeasurable/unknown  disturbances $w_\mr{u}$, such as measurement noise and load variation, composing $w = \col(w_\mr{m},w_\mr{u})$.  Hence, to guarantee convergence towards the designed desired steady-state trajectory, we also require knowledge of the asymptotic behavior of the unknown part $w_\mr{u}$ of $\st{w}$. We only need knowledge on the asymptotic behavior of $w_\mr{u}$ as only that influences the steady-state trajectory  $\stTraj$, e.g. a zero mean measurement noise does not need to be estimated as it does not influence the steady-state trajectory, while estimation of a constant load is required as it directly influences it. %the steady-state trajectory.
%Hence, we also require knowledge of the corresponding disturbance $\st{w}$ to guarantee convergence towards the steady-state trajectory. In the generalized plant framework, $w$ (and hence $\st{w}$) often consists of known disturbances, e.g. reference trajectories, and unknown disturbances, e.g. in- and output disturbances. 

Disturbance observers have been widely used to estimate and compensate for the effect of unknown disturbances \cite{Chen2016}. Often they work on the basis of the internal model principle, whereby the (assumed) dynamics of the disturbance are included in the design \cite{Chen2004}. Similarly, in this work, we make use of a disturbance observer in order to estimate the unknown elements $\st{w_\mr{u}}$ of the steady-state generalized disturbance $\st{w}$. % More specifically, we assume that the 
%\TR{Would be nice to distinguish the known and the unknown part.}

%\TR{Also we need to mention that we need to track the disturbance approximately as only the convergence point matters, i.e., we need to establish that to which steady state orbit we are converging to.}

\begin{asum}[Disturbance model]\label{asum:distmodel}
	Given the generalized plant \eqref{eq:genplantfull}, assume that the unknown disturbances $w_\mr{u}$, can be modeled by the disturbance generator
		\begin{equation}
	\left\lbrace
	\begin{aligned}
		\dot{x}_\mr{w}(t) &= f_\mr{w}(x_\mr{w}(t));\\
		w_\mr{u}(t) &= C_\mr{w}x_\mr{w}(t);
	\end{aligned}\right.
	\end{equation}
	with $x_\mr{w}(t)\in\mathbb{X}_\mr{w}$.
\end{asum}

Given Assumption \ref{asum:distmodel}, the state dynamics of the combined generalized plant \eqref{eq:genplantfull} and disturbance model are given by
\begin{subequations}\label{eq:distmodel}
\begin{equation}
	\begin{aligned}
		\dot{x}_\mr{e}(t) &=\begin{bmatrix}
			\dot x_\mr{\pnote}(t)\\\dot x_\mr{w}(t)
		\end{bmatrix} = \begin{bmatrix}f_\mr{\pnote}(x_\mr{\pnote}(t),\genInput(t))+B_\mr{w}
		\begin{bmatrix} w_\mr{m}(t)\\C_\mr{w}x_\mr{w}(t)\end{bmatrix}\\
		f_\mr{w}(x_\mr{w}(t))
		\end{bmatrix},\\
		&= f_\mr{e}(x_\mr{e}(t),\genInput(t),w_\mr{m}(t)),\\
	\end{aligned}
\end{equation}
where $x_\mr{e}(t)\in\mathbb{X}_\mr{\pnote}\times\mathbb{X}_\mr{w}=\mathbb{X}_\mr{e}\subseteq\mathbb{R}^{n_\mr{x_e}}$, 
and the measured output is given by 
\begin{equation}
\begin{aligned}
	\genOutput &= C_\mr{y}x_\mr{\pnote}(t)+D_\mr{yw}\begin{bmatrix}w_\mr{m}(t)\\C_\mr{w}x_\mr{w}(t)\end{bmatrix},\\
	&= h_\mr{e}(x_\mr{e}(t),w_\mr{m}(t)).
	\end{aligned}
\end{equation}
\end{subequations}

\begin{rem}
	Note that the disturbance model in Assumption \ref{asum:distmodel} has a different purpose than the disturbance model that is included in the generalized plant $P$ given by \eqref{eq:genplant}, which is modeled in terms of weighted norm relation. Namely, the former is used to model disturbances that are acting on the system and influence the asymptotic behavior of the steady-state trajectory (and hence influence $\st{w}$), while the latter can be seen as modeling the difference/increment between the steady-state disturbance and other possible disturbances (i.e. $w-\st{w}$) and disturbance which do not influence the asymptotic behavior of the steady-state trajectory, e.g., measurement noise.
\end{rem}

\begin{thm}[Nonlinear observer]\label{thm:nlobsv} %{[Reshuffling of theorem]}
The state observer
\begin{equation}\label{eq:observer}
	\left\lbrace
	\begin{aligned}
		\dot{\hat x}_\mr{e}(t) &= f_\mr{e}(\hat{x}_\mr{e}(t),\genInput(t),w_\mr{m}(t))+L(y(t)-\hat{y}(t));\\
		\hat{y}(t) &= h_\mr{e}(\hat{x}_\mr{e}(t),w_\mr{m}(t));
	\end{aligned}\right.
	\end{equation}
	with $\hat{x}_\mr{e}(t)\in\mathbb{R}^{n_\mr{x_e}}$ and $\hat{y}(t)\in\mathbb{R}^{n_\mr{y}}$ %is a state observer for \eqref{eq:distmodel}, in the sense 
	ensures that for $t\rightarrow\infty$, $\hat{x}_\mr{e}(t)\rightarrow x_\mr{e}(t)$, if there exists a $P\succ 0$ with $P\in\mathbb{R}^{n_\mr{x_e}\times n_\mr{x_e}}$ and an $F\in\mathbb{R}^{n_\mr{x_e}\times n_\mr{y}}$ such that
	\begin{multline}\label{eq:obsvLMI}
		\m A_\mr{e}(x_\mr{e},\genInput,w_\mr{m})^\top P+P\m \m A_\mr{e}(x_\mr{e},\genInput,w_\mr{m})-\\\m C_\mr{e}(x_\mr{e},w_\mr{m})^\top F-F\m C_\mr{e}(x_\mr{e},w_\mr{m})\prec0,
	\end{multline}
	for all $(x_\mr{e},\genInput,w_\mr{m})\in\mathbb{X}_\mr{e}\times\genInputSet\times\pi_\mr{w_\mr{m}}\mathbb{W}$, where $A_\mr{e} = \frac{\partial f_\mr{e}}{\partial x_\mr{e}}$, $C_\mr{e} = \frac{\partial h_\mr{e}}{\partial x_\mr{e}}$, and $L = P F^{-1}$.
\end{thm}
\begin{proof}
	See Appendix \ref{App:E}.
\end{proof}
\begin{rem}
	Similar to controller design, also for observer design the LPV framework can be used by embedding the differential form of the system in an LPV representation. This allows to use convex optimization for the computation of the observer gain $L$ in Theorem \ref{thm:nlobsv}.
\end{rem}

Applying the nonlinear observer \eqref{eq:observer} gives us a disturbance observer for the combined generalized plant and disturbance model \eqref{eq:distmodel}, which then allows us to estimate the unknown disturbances on the system by taking $\st{w}_\mr{u}(t) = C_\mr{w}\hat{x}_\mr{w}(t)$, where $\hat{x}_\mr{w}$ is the estimated state of the disturbance model. This can then be used to compute the steady-state control input trajectory $\st{\genInput}=\st{y_\mr{k}}$, used by the realized controller $K$ \eqref{eq:Contr}, corresponding to the steady-state output trajectory $\st{\genOutput}$. Note that co-design of the controller and the observer under \eqref{eq:observer} can also be accomplished. %However, further investigation into this is outside of the scope of the paper.

\subsubsection{Unknown steady-state solution}
In case $\st{w}$ cannot be measured or estimated, we cannot guarantee that $\st{w}(t)\rightarrow w(t)$ as $t\rightarrow\infty$. Consequently, we cannot ensure convergence towards our desired steady-state solution $\stTraj$. However, in this case, the controller still ensures an \ltwo-gain bound\footnote{Assuming that $w-\st{w}$ is in the extended \ltwo-space.} of $\gamma$ from $w-\st{w}$ to $z-\st{z}$, i.e., the steady state trajectory will remain close in an \ltwo sense to the desired reference and the controller will still ensure stability of the closed-loop system. This weaker performance guarantee is also referred to as universal \ltwo-gain performance, see also \cite{Manchester2018}

%we still have \ltwo-gain performance and the controller will still stabilize the system, as , . 

%%%%%%%%%%%%%%%%%%%%%%%%%%%%%%%%%%%%%%%%%%%%%%%%%%%%%%%%%%%%%%%%%%%%%%%%%%%%%%%%%%%%%%%%%%%%%%%%%
%%%%% Examples %%%%%
\section{Examples}\label{sec:Examples}
In this section, \litwo-gain stability and performance of the closed-loop NL system guaranteed by the proposed incremental controller will be demonstrated through examples and compared to standard LPV controller designs which only ensure \ltwo-gain stability and performance. We will also demonstrate using simulation and experimental results that standard LPV control can fail to result in the desired behavior, while the proposed LPV synthesis resulting in an incremental controller can reliably achieve it.%
\subsection{Duffing Oscillator}\label{sec:duffspring}
First, reference tracking with a Duffing oscillator is investigated. The system is described by the following differential equations:
\begin{equation}
\begin{aligned}
        \dot{q}(t) &= v(t);\\
        \dot{v}(t) &= -\frac{k_1}{m} q(t) -\frac{k_2}{m} \left(q(t)\right)^3 - \frac{d}{m} v(t) + \frac{1}{m} F(t);\label{eq:duffNL}
       \end{aligned}
\end{equation}
where, $q$ $[\mr{m}]$ is the position, $v$ $[\mr{m\cdot s^{-1}}]$ the velocity and $F$ $[\mr{N}]$ is the (input) force acting on the mass. Furthermore, $m = 1\;[\mathrm{kg}]$, $k_1 = 0.5\;[\mathrm{N\cdot m^{-1}}]$, $k_2 = 5\;[\mathrm{N\cdot m^{-3}}]$ and $d = 0.2\;[\mathrm{N\cdot s\cdot m^{-1}}]$. Only the position $q$ is assumed to be measurable and hence it is considered to be the only output of the plant.

For comparison, besides designing an incremental controller, a standard LPV controller design is also made. For the design of the standard LPV controller, the primal form of the system \eqref{eq:duffNL} is embedded in an LPV representation, where dependence on time is omitted for brevity:
\begin{equation}\label{eq:duffLPV}
    \begin{aligned}
        \dot{q} &= v;\\
        \dot{v} &= \left(-\frac{k_1}{m}-\frac{k_2}{m}\sche_\mathrm{s}\right) q - \frac{d}{m} v + \frac{1}{m} F;
    \end{aligned}
\end{equation}
where $\sche_\mathrm{s}(t) = q^2(t)$ is the scheduling-variable. Here we will denote with subscript `s' this `standard' concept of LPV embedding and control design. $\mathcal{P}_\mathrm{s}=[0,\, 2]$ is chosen to allow for a relatively large operating range. No rate bounds are assumed on $\sche_\mathrm{s}$. For the incremental  synthesis, the differential form of \eqref{eq:duffNL} is calculated and is embedded in an LPV representation, resulting in
\begin{equation}\label{eq:duffincr}
    \begin{aligned}
        \delta \dot{q} \! &= \delta {v};\\
        \delta \dot{v} \! &= \! \left(\!-\frac{k_1}{m}-3 \frac{k_2}{m} \sche\!\right)\! \delta q - \frac{d}{m} \delta v + \frac{1}{m} \delta F;
    \end{aligned}
\end{equation}
where the scheduling $\sche(t) = q^2(t)$ results to be the same as in \eqref{eq:duffLPV} and hence it is also considered with $\mathcal{P}=[0,\, 2]$.

The generalized plant $P$ used for synthesis is depicted in Fig. \ref{fig:genplantscor} where $G$ is the system given by \eqref{eq:duffNL}, $K$ is the controller, $w=\col\left( r,d_\mathrm{i}\right)$ is the generalized disturbance with $r$ the reference and $d_\mathrm{i}$ being an input disturbance. The performance channel consists of $z_1$ (tracking error) and $z_2$ (control effort). The considered LTI weighting filters $\lbrace W_i \rbrace_{i = 1}^3$ are chosen as the transfer functions $W_1(s) = \frac{0.501(s + 3)}{s + 2\pi}$, $W_2(s) = \frac{10(s+50)}{s + 5\cdot 10^4}$ and $W_3 = 1.5$ where $s$ corresponds to the complex frequency. Furthermore, integral action is enforced by the filter $M(s) = \frac{s+2\pi}{s}$. The resulting sensitivity weight $W_1(s)M(s)$ has guaranteed 20 dB/dec roll-off at low frequencies in order to ensure good tracking performance, while $W_2(s)$ has high-pass characteristics in order to ensure proper roll-off at high frequencies. Because the system \eqref{eq:duffincr} and the corresponding generalized plant have affine dependency on the scheduling-variable, polytopic \ltwo-gain synthesis based on \cite{Apkarian1998,Apkarian1995} has been used in the design of both the incremental and standard LPV controllers to ensure closed-loop \litwo and \ltwo-gain stability and performance, respectively. This synthesis algorithm has been implemented in the LPVcore Toolbox \cite{DenBoef2021}, which has been used to synthesize the controllers.

\begin{figure}
    \centering
    \includegraphics[scale=.9]{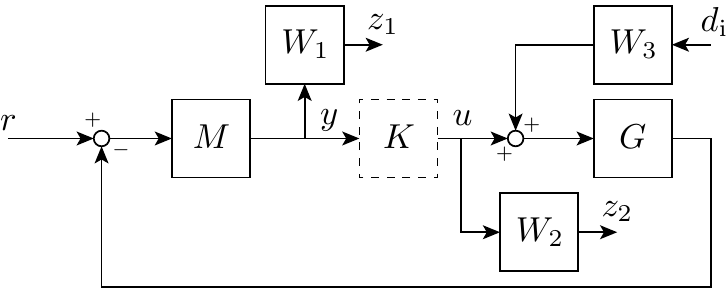}
    \caption{Generalized plant used in the examples.}\vspace{-1em}
    \label{fig:genplantscor}
\end{figure}

Using the given weighting filters, synthesizing the standard LPV controller results in an \ltwo-gain of $\gamma \approx 0.94$ and synthesizing the LPV controller for the differential plant results in an \litwo-gain of $\gamma \approx 1.2$. As the synthesis provides both controllers with affine parameter dependence, we can use  in the \litwo-case the result of Corollary \ref{cor:affineint} in order to compute $\intsche(t) = \int_0^1\psi(\bar{x}_\mr{\pnote}(t,\lambda)\,d\lambda=q^2(t)+q(t)\st{q}(t)+(\st{q}(t))^2$ and realize the incremental controller $K$. As we assume the presence of an (unknown) input disturbance $d_{\mathrm{i}}$,  a disturbance observer is constructed for the incremental controller as described in Section \ref{sec:distobsv}. It is assumed the input disturbance will be constant, hence, the following disturbance model is used for the disturbance observer design
\begin{equation}\label{eq:constdistmodel}
	\begin{aligned}
		\dot x_\mr{d}(t) &= 0;\\
		d_\mr{i}(t) &= x_\mr{d}(t).
	\end{aligned}
\end{equation}
\begin{figure}
    \centering
    \includegraphics[scale=1]{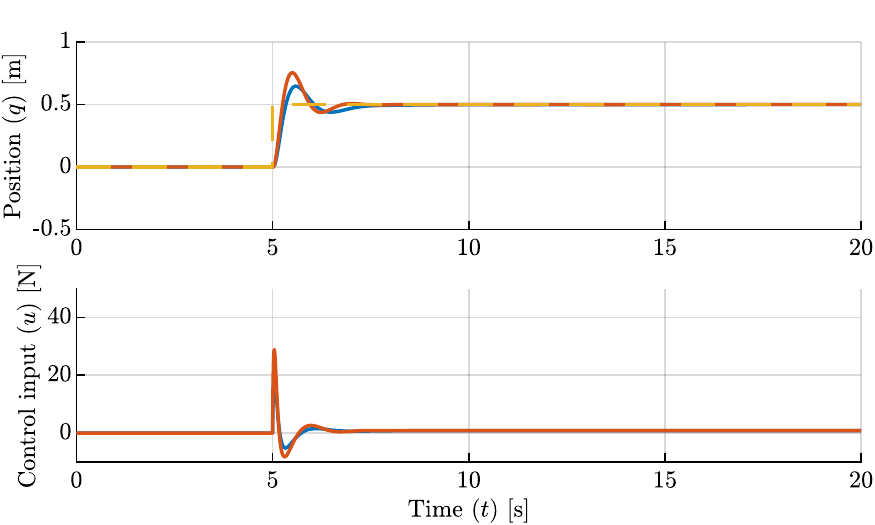}
    \caption{Position of the Duffing oscillator (top) in closed-loop with the standard LPV (\legendline{mblue}) and the  \linebreak incremental (\legendline{morange}) controllers under reference (\legendline{myellow,dashed}) and no input disturbance, together with the generated control inputs (bottom) by the controllers.}
    \label{fig:duffNoDist}
\end{figure} 
\begin{figure}
    \centering
    \includegraphics[scale=1]{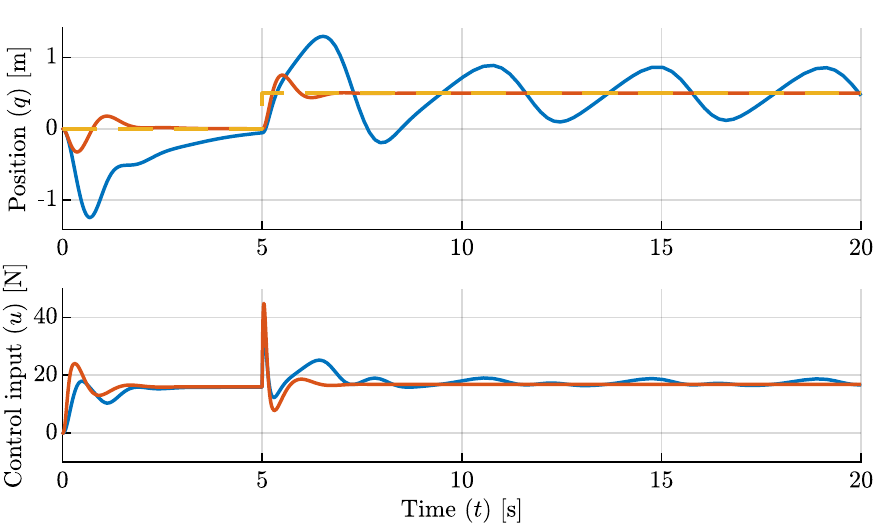}
    \caption{Position of the Duffing oscillator (top) in closed-loop with the standard LPV (\legendline{mblue}) and the \linebreak incremental (\legendline{morange}) controllers under reference (\legendline{myellow,dashed}) and input disturbance, together with the  generated control inputs (bottom) by the controllers.}
    \label{fig:duffDist}
\end{figure}
In simulation, the resulting outputs of the system using the standard LPV controller and the incremental controller in closed-loop are depicted without and with input disturbance in Fig. \ref{fig:duffNoDist} and Fig. \ref{fig:duffDist}, respectively. In both cases, a step signal is taken as a reference trajectory which changes from zero to 0.5 at $t=5$ seconds. For the incremental controller, the reference $r$ corresponds to $\st{q}$ with $\st{\genInput}(t)=k_1 \st{q}(t)+k_2 (\st{q}(t))^3-W_3 \st{d_i}(t)$ (as the trajectory of $\st{q}$ is piecewise constant). For the simulation results in Fig. \ref{fig:duffDist}, a constant input disturbance $d_\mathrm{i} \equiv -10\tfrac{2}{3}$ (corresponding to $-10\tfrac{2}{3}\cdot W_3=-16$ [N]) is applied. Comparing the results of the standard and the incremental controllers in Fig. \ref{fig:duffNoDist} shows that both controllers have similar performance when no input disturbance is present. The incremental controller has slightly more overshoot, but a lower settling time for this example. However, under constant input disturbance, it can be seen in Fig. \ref{fig:duffDist} that the standard LPV controller has a significant performance loss with oscillatory behavior, whereas the incremental controller preserves its smooth reference tracking property. Note, that in both cases, the scheduling variable $\sche$ never leaves the set for which the controllers have been designed, i.e. $q^2(t) \in [0,2]$. 

%%%%%%%%%%%%%%%%%%%%%%%%%%%%%%%%%%%%%%%%%%%%%%%%%%%%%
\subsection{Unbalanced disk}\label{sec:unbdisc}
\begin{figure}
	\centering
	\includegraphics[width=0.8\columnwidth]{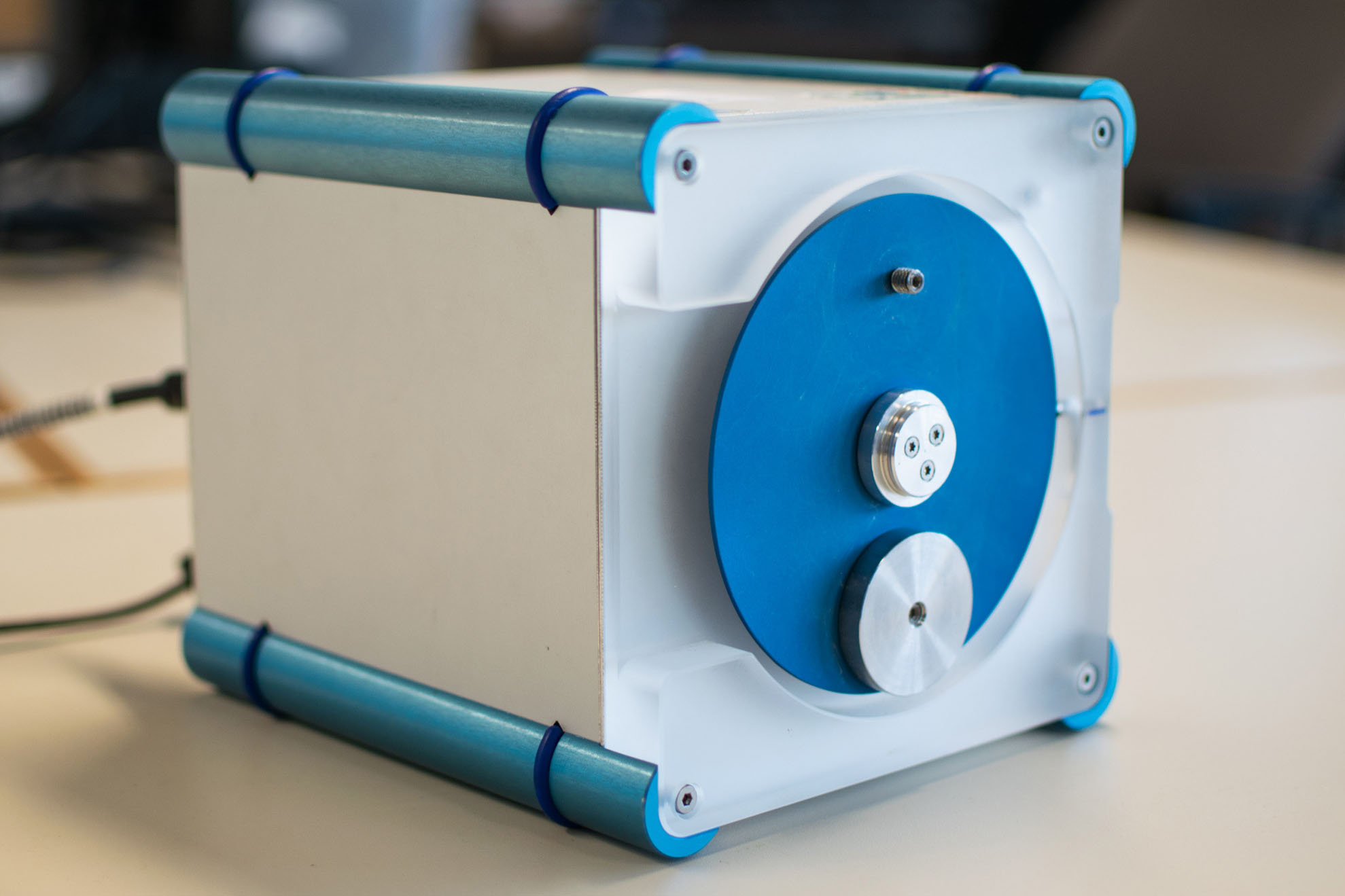}
	\caption{Unbalanced disk setup.}\label{fig:ubdisk}
\end{figure}

For the next example, the proposed control method is experimentally verified on an unbalanced disk setup given in Fig. \ref{fig:ubdisk}. By neglecting the fast electrical dynamics of the motor, the motion of this system can be described (see,  \cite{Kulcsar2009}) as
\begin{equation}\label{eq:disc}
\begin{aligned}
    \dot{\theta}(t) &= \omega(t);\\
    \dot{\omega}(t) &= \tfrac{M g l}{J}\sin(\theta(t)) -\tfrac{1}{\tau}\omega(t)+\tfrac{K_m}{\tau}V(t);
\end{aligned}
\end{equation}
where $\theta$ $[\mr{rad}]$ is the angle of the disk, $\omega$ $[\mr{rad\cdot s^{-1}}]$ its angular velocity and $V$ $[\mr{V}]$ is the input voltage to the motor. The angle of the disk $\theta$ is considered to be the output of the plant.
The physical parameters, estimated based on measurement data, are given in Table \ref{table:disc}.% see Appendix \ref{sec:ParamEst} for a brief description of the physical parameter estimation.
\begin{table}
\centering
\caption{Physical parameters of the unbalanced disk.}
\label{table:disc}
\begin{tabular}{c||c|c|c|c|c|c}
Parameter & $g$ & $J$ & $K_m$ & $l$ & $M$ & $\tau$ \\\hline
Value     & 9.8 & $2.4\mathrm{e-}4$ & 11  & 0.041      & 0.076      & 0.40    \\ 
\end{tabular} 
\end{table}
The differential form of \eqref{eq:disc}, embedded in an LPV representation, is given by
\begin{equation}\label{eq:discincrLPV}
\begin{aligned}
    \delta \dot{\theta}(t) &= \delta \omega(t);\\
    \delta \dot{\omega}(t) &= \left(\tfrac{M g l}{J}\sche(t))\right)\delta \theta(t) -\tfrac{1}{\tau}\delta \omega(t)+\tfrac{K_m}{\tau}\delta V(t);
\end{aligned}
\end{equation}
where $\sche(t)=\psi(\theta(t))=\cos(\theta(t))$ is the scheduling-variable which is assumed to be in $\mathcal{P}=[-1,\, 1]$. 
%Note that $\dot{\sche}(t) = -\sin(\theta(t))\omega(t)$ for which no bounds are explicitly assumed.

As in the previous example, the primal form of the NL system  \eqref{eq:disc} is also embedded in an LPV representation for the sake of comparison, which results in
\begin{equation}\label{eq:discLPV}
\begin{aligned}
    \dot{\theta}(t) &= \omega(t);\\
    \dot{\omega}(t) &= \left(\tfrac{M g l}{J}\sche_\mathrm{s}(t))\right) \theta(t) -\tfrac{1}{\tau}\omega(t)+\tfrac{K_m}{\tau}V(t);
\end{aligned}
\end{equation}
where $ \sche_\mathrm{s}(t) =\psi_\mathrm{s}(\theta(t)) = \tfrac{\sin(\theta(t))}{\theta(t)} = \mr{sinc}(\theta(t))$. 
$\mathcal{P}_\mathrm{s}$ is chosen\footnote{Note that $\psi_\mr{s}(0) = 1$ as $\lim_{x\rightarrow 0}\mr{sinc}(x)=1$.} as $[-0.22,\, 1]$ with no assumptions on the rate bounds. The used generalized plant structure is depicted in Fig. \ref{fig:genplantscor}. The weighting filters are chosen as
$W_\mathrm{1}(s) = \frac{0.5012( s + 4)}{s+\pi}$, $M(s) =\frac{s+\pi}{s}$, $W_\mathrm{2}(s) = \frac{s+40}{s+4000}$ and $W_\mathrm{3} = 0.5$.
%Note that the integral action is approximate in this case due to the choice of $W_\mr{s}$, as $W_\mr{s}$ includes a real pole close the origin.
%\todo{[Todo: add sentence about approximate integral action in this example]}
Synthesizing the controllers, using the same approach as for the duffing oscillator in Section \ref{sec:duffspring}, results in an \ltwo-gain of $\gamma \approx 1.1$ and \litwo-gain of $\gamma \approx 1.2$. %As the control strategies are implemented on the experimental setup, the controllers had to be discretized. For discretization of the LPV controllers, the complete LPV-SS discretization method from \cite{VandenHof2010} was implemented online on the controller with a sampling time of 5 ms, which is the lowest sampling time this platform supports. 
As the LPV controller resulting for the differential form of the plant has an affine dependency, we can use Corollary \ref{cor:affineint} to compute\footnote{Note that $\lim_{\theta\rightarrow\st{\theta}}\frac{\sin(\theta)-\sin(\st{\theta})}{\theta-\st{\theta}}=\cos(\theta)=\cos(\st{\theta})$.} $\intsche(t) = \int_0^1\psi(\bar{x}_\mr{\pnote}(t,\lambda)\,d\lambda=\frac{\sin(\theta(t))-\sin(\st{\theta}(t))}{\theta(t)-\st{\theta}(t)}$. For the resulting incremental controller, a disturbance observer is also designed to estimate the (unknown) disturbances $d_\mr{i}$. As  $d_\mr{i}$ is assumed to be constant, \eqref{eq:constdistmodel} is used for the design.
On the experimental setup, for safety, the input voltage to the system is saturated between $\pm$ 10 [V].

In Fig. \ref{fig:disc_exp_nodist}, the measured angular response of the disc during the experiments is depicted along with the input to the setup (i.e. $V$). The reference trajectory $r$ switches between $0$ and $\pm \frac{\pi}{2}$ rad/s. For the incremental controller, $r$ corresponds to $\st{\theta}$ with $\st{\genInput}(t)=-\frac{M g l \tau}{J K_m}\sin(\st{\theta}(t))-W_3\st{d_\mr{i}}(t)$ (as the trajectory is piecewise constant).

In Fig. \ref{fig:disc_exp_dist}, the same reference trajectory is used, but a constant input disturbance of $d_\mathrm{i}= 100$, corresponding to $100\cdot W_\mr{3}=50$ [V], is introduced (which is implemented by adding 50 V to the control input that is sent to the plant before saturation). Note that the reference only starts at 10s to give the controllers time to compensate for the input disturbance. The standard LPV controller performs much worse when a constant input disturbance is present, compared to the incremental controller, which has similar performance to the case when no input disturbance is applied. Both the standard LPV and the incremental controllers are able to compensate the 50 [V] input disturbance, as visible in the total received input by the plant (i.e. $V=\genInput+W_3 d_\mr{i}$), see bottom graph in Fig. \ref{fig:disc_exp_dist}. However, while the input that is sent to the plant is nearly identical for the incremental controller in both cases, see Fig. \ref{fig:disc_exp_nodist} and Fig. \ref{fig:disc_exp_dist}, this is clearly not the case for the standard LPV controller. For the latter, oscillations in the input signal are present when the input disturbance is applied which causes unwanted oscillation of the disk angle. While an input disturbance of 50 [V] is extraordinarily high for this system, and will likely never occur on the real setup, it still shows that there are inherent issues when using  standard LPV controllers. 
\begin{figure}
    \centering
    \includegraphics[scale=1]{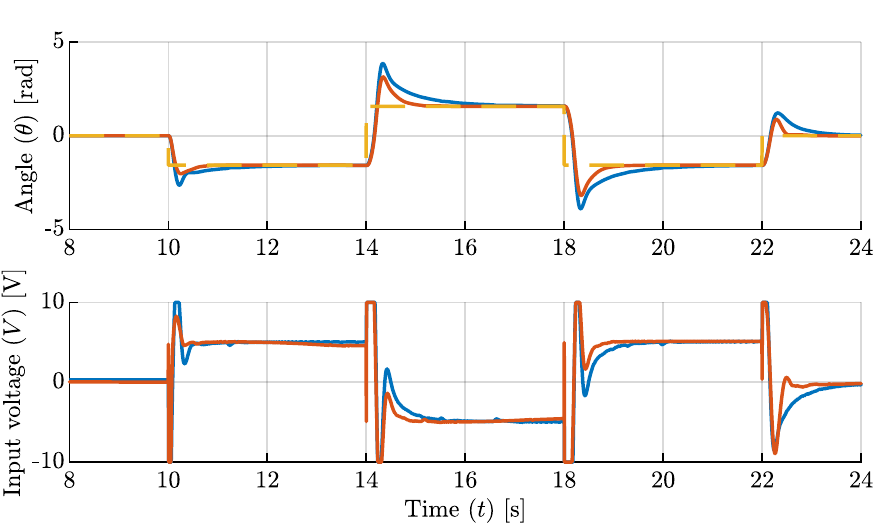}
    \caption{Measured angle of the unbalanced disk system (top)  in closed-loop with the standard LPV (\legendline{mblue}) and the incremental \\(\legendline{morange}) controllers under reference \mbox{(\legendline{myellow,dashed})} and no input disturbance, together with inputs to the plant (bottom) generated by the controllers.}    \label{fig:disc_exp_nodist}
\end{figure}
\begin{figure}
    \centering
    \includegraphics[scale=1]{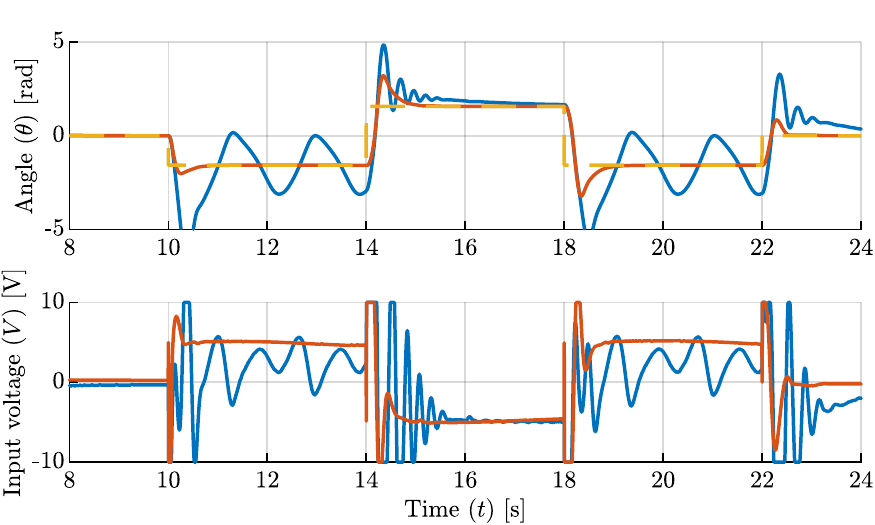}
    \caption{Measured angle of the unbalanced disk system (top) in closed-loop with the standard LPV  (\legendline{mblue}) and the incremental \linebreak(\legendline{morange}) controller under reference \mbox{(\legendline{myellow,dashed})} and input disturbance, together with corresponding inputs to the plant (bottom) generated by the controllers.}
    \label{fig:disc_exp_dist}
\end{figure}

%%%%%%%%%%%%%%%%%%%%%%%%%%%%%%%%%%%%%%%%%%%%%%%%%%%%%
\extver{
\subsection{Scorletti et al. example}\label{sec:scorex}
Finally, we compare the results of our method with the results from \cite{Scorletti2015}. The example system in \cite{Scorletti2015} is described by the following state-space representation
\begin{equation}\label{eq:scorsys}
\begin{aligned}
	\dot{x}_\mr{g,1}(t) &= -100\varphi(x_\mr{g,1}(t))-70x_\mr{g,2}(t)+300u_\mr{g}(t);\\
	\dot{x}_\mr{g,2}(t) &= 70 x_\mr{g,1}(t)-14x_\mr{g,2}(t);\\
	y_\mr{g}(t) &= x_\mr{g,1}(t);
\end{aligned}
\end{equation}
where
\begin{equation}
	\varphi(x) = \left\lbrace
	\begin{alignedat}{2}
	&0.9x^3-2\vert x\vert x+1.2 x, \; &&\text{for}\; \vert x \vert < \tfrac{5}{3};\\
	&2x-2.72, \; &&\text{for}\; x\geq \tfrac{5}{3};\\
	&2x+2.72, \; &&\text{for}\; x\leq -\tfrac{5}{3}.
	\end{alignedat}
	\right.
\end{equation}
Computing the differential form of \eqref{eq:scorsys} and embedding it in an LPV representation results in
\begin{align}\label{eq:scordifsys}
\begin{aligned}
	\delta \dot{x}_\mr{g,1}(t) &= \left(-100\sche(t))\right)\delta x_\mr{g,1}-70\delta x_\mr{g,2}(t)+300\delta u_\mr{g}(t);\\
	\delta \dot{x}_\mr{g,2}(t) &= 70 \delta x_\mr{g,1}(t)-14\delta x_\mr{g,2}(t);\\
	\delta y_\mr{g}(t) &= \delta x_{\mr{g},1}(t);
\end{aligned}\raisetag{13pt}
\end{align}
where $\sche(t) = \varphi_{\delta}(x_\mr{g,1}(t))$, is the scheduling-variable, which is assumed to be in $\mathcal{P}\in [-0.3 ,2]$, where
\begin{equation}
	\varphi_{\delta}(x) = \left\lbrace
	\begin{alignedat}{2}
	&2.7x^2-4x+1.2, \; &&\text{for}\; 0\leq x < \tfrac{5}{3};\\
	&2.7x^2+4x+1.2, \; &&\text{for}\; -\tfrac{5}{3} < x < 0;\\
	&2, \; &&\text{for}\; x\geq \tfrac{5}{3} \wedge x\leq -\tfrac{5}{3}.
	\end{alignedat}
	\right. 
\end{equation}
A generalized plant structure is taken as in Fig. \ref{fig:genplantscor}, with $W_1(s) = \frac{50}{s+2\pi}$, $W_2(s) = \frac{10(s+10)}{s+1000}$, $W_3=0.1$ and $M(s) = \frac{s+2\pi}{s}$. Which is similar to the generalized plant and weighting filters taken in \cite{Scorletti2015}. 
On the basis of this, an incremental controller for \eqref{eq:scorsys} is synthesized, using \eqref{eq:scordifsys}. Like for previous examples, the method from \cite{Apkarian1998} is used during the synthesis procedure. This results in the incremental controller achieving an \litwo-gain of $\gamma \approx 1.0$, similar to the incremental gain obtained in \cite{Scorletti2015}, where an \litwo-gain of $\gamma \approx 1$ is reported. 
As the LPV controller resulting for the differential form of the plant has affine dependency we can, like was done for the previous examples, use the result of Corollary \ref{cor:affineint} in order to compute\footnote{Note that $\lim_{x_1\rightarrow \st{x_1}}\frac{\varphi(x_\mr{g,1})-\varphi(\st{x_\mr{g,1}})}{x_\mr{g,1}-\st{x_\mr{g,1}}}=\varphi_\delta(x_\mr{g,1})=\varphi_\delta(\st{x_\mr{g,1}})$.} $\intsche(t) = \int_0^1\bar{\sche}(t,\lambda))\,d\lambda=\frac{\varphi(x_\mr{g,1}(t))-\varphi(\st{x_\mr{g,1}}(t))}{x_\mr{g,1}(t)-\st{x_\mr{g,1}}(t)}$. Furthermore, like for the previous examples, a disturbance observer is used for the incremental controller for which the disturbance model \eqref{eq:constdistmodel} is also used. In order to compute the feasible steady-state trajectory used by the incremental controller, the reference $r$ is filtered by the lowpass filter $F(s)=\frac{1000}{s+1000}$ (where $s$ is the complex frequency) which then corresponds to $\st{x_\mr{g,1}}$, this trajectory is then used to compute the corresponding control input $\st{\genInput}(t)$ (which due to its complexity is not given). Fig. \ref{fig:scorincrtime} shows the output of the NL system from \cite{Scorletti2015} in closed-loop with the proposed incremental controller alongside the results from \cite[Fig. 7]{Scorletti2015}, as well as the input to the plant, i.e. $u_\mr{g}$, for the proposed incremental controller.

It can be observed that our proposed incremental control design performs better in this example than the one proposed in \cite{Scorletti2015}, using the same performance and stability requirements set by the weighting filters. This is likely due to the fact that (i) our proposed controller has a more flexible dependency structure for the scheduling-variable, compared to the linear structure of the controller proposed in \cite{Scorletti2015}; (ii) our proposed controller contains besides a feedback part also a feedforward component, corresponding to the steady-state trajectory, which the method in \cite{Scorletti2015} does not have.
\begin{figure}
    \centering
    \includegraphics[scale = 1]{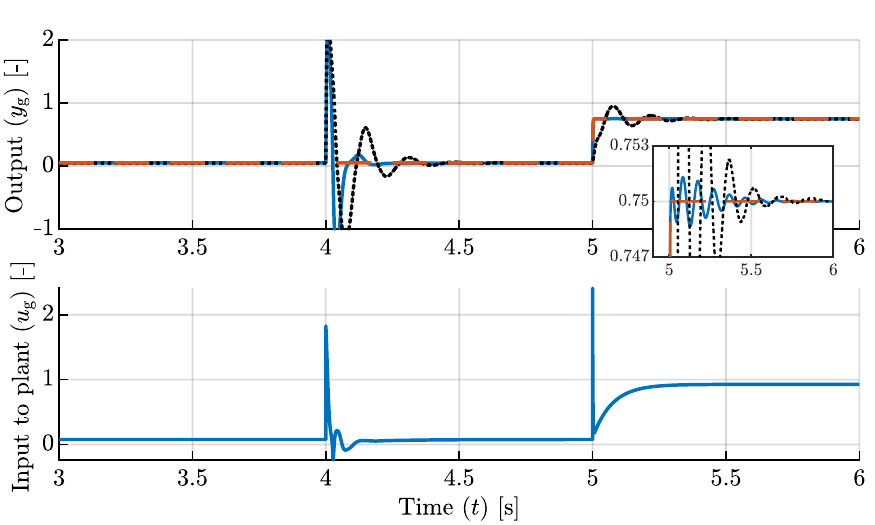}
    \caption{Output of the NL system \eqref{eq:scorsys} (top) in closed-loop with the proposed incremental controller (\legendline{mblue}), as well as the results from \cite[Fig. 7]{Scorletti2015} (\legendline{black,dotted}) under the reference trajectory \linebreak(\legendline{morange,dashed}), together with the input to the plant (bottom).}
    \label{fig:scorincrtime}
\end{figure}}
{In the extended version of the paper \cite{Koelewijn2022} we also apply the proposed incremental controller design to the example system from \cite{Scorletti2015} for comparison purposes.}

%%%%%%%%%%%%%%%%%%%%%%%%%%%%%%%%%%%%%%%%%%%%%%%%%%%%%%%%%%%%%%%%%%%%%%%%%%%%%%%%%%%%%%%%%%%%%%%%%
%%%%% Conclusion %%%%%
\section{Conclusion}\label{sec:Conclusion}
In this paper, we proposed a novel systematic dynamic output feedback controller synthesis method for nonlinear systems under general controller parameterization which provides incremental stability and performance guarantees of the achieved closed loop behavior. The proposed incremental controller synthesis method is based on two key ingredients: (i) linear parameter-varying (LPV) controller synthesis on the differential form of the nonlinear plant to be controlled and (ii) realization of the resulting LPV controller as an implementable incremental nonlinear controller with closed-loop stability and performance guarantees. A key advantage of the method that it facilitates systematic controller design for nonlinear plants by convex synthesis and enables the use of powerful performance shaping concepts available for linear controller design. Although a large variety of quadratic dissipativity notions can be ensured by the proposed methodology, we chose to exemplify the approach with incremental \ltwo-gain and performance. As it is demonstrated through simulation and experimental studies, the proposed approach successfully achieves desired closed-loop stability and performance requirements for tracking and rejection problems and overcomes issues of standard LPV controller synthesis methods. For future research, we aim at further increase achievable performance of the approach by using state/parameter-dependent storage functions.

%%%%%%%%%%%%%%%%%%%%%%%%%%%%%%%%%%%%%%%%%%%%%%%%%%%%%%%%%%%%%%%%%%%%%%%%%%%%%%%%%%%%%%%%%%%%%%%%%
%%%%% Appendices %%%%%
\appendices
\section{Proof of Theorem \ref{th:2}}\label{App:proofdiffincr}
Given a $V_\Delta\in\m{C}_1$, \eqref{eq:IDIE} holds if and only if (see \cite{Willems1972}) 
	\begin{equation}\label{eq:idddie}
		\begin{aligned}
				\frac{\partial V_\Delta}{\partial t}(x(t),\tilde{x}(t)) = &\frac{\partial V_\Delta}{\partial x}(x(t),\tilde{x}(t))f(x(t),w(t))+\\ &\frac{\partial V_\Delta}{\partial \tilde{x}}(x(t),\tilde{x}(t))f(\tilde{x}(t),\tilde{w}(t))\\ &\hspace{1em} \leq s_\Delta(w(t),\tilde{w}(t),z(t),\tilde{z}(t)),
		\end{aligned}
	\end{equation}
	for all $t\in\mathbb{R}_+$ and all solutions $(x,w,z),(\tilde{x},\tilde{w},\tilde{z})\in\mathfrak{B}$. As $(x(t),w(t)) \in \mathbb{X}\times\mathbb{W}$ and $(\tilde{x}(t),\tilde{w}(t)) \in \mathbb{X}\times\mathbb{W}$ for all $t\in\mathbb{R}$, \eqref{eq:idddie} holds if \eqref{eq:DDIE} holds.
	
\section{Proof of Theorem \ref{thm:incrstab}}\label{App:incrstab}
The storage function $V_\Delta$ defined in Definition \ref{def:incrdissip} satisfies the conditions for a (incremental) Lyapunov function in \cite[Theorem 1]{Angeli2002}. Furthermore, \cite[Eq.~(9)]{Angeli2002} is satisfied by \eqref{eq:IDIE} under the assumptions on the supply function in \eqref{eq:quadsupp} and \eqref{eq:suppassum}.

\section{Proof of Lemma \ref{lem:convergence}}\label{App:convergence}
Lemma \ref{lem:convergence} simply follows from \cite[Theorem 11]{Ruffer2013}
as $\Sigma$ is incrementally stable if $\Sigma$ is incrementally dissipative for a supply function of the form \eqref{eq:quadsupp} which satisfies \eqref{eq:suppassum}, see Theorem \ref{thm:incrstab}. Moreover, as $\mathcal{X}$ is compact and %positive 
invariant, by \cite[Theorem 11]{Ruffer2013}, this implies that $\Sigma$ is uniformly convergent on $\mathcal{X}$ for any bounded $w^* \in \pi_\mr{u}\mathfrak{B}^\mathcal{X}$ to a steady-state solution $x^* \in \pi_\mr{x}\mathfrak{B}^\mathcal{X,W}$. Hence, for any solution $(x,w^*)\in \mathfrak{B}^\mathcal{X,W}_\mr{x,w}$, $\lim_{t\rightarrow\infty}(x(t)-x^*(t))=0$. 

\section{Proof of Lemma \ref{lem:incrl2gainIncrDissip}}\label{App:incrl2gaindissip}
See \cite{Verhoek2020} for a proof that incremental dissipativity of $\Sigma$ with the supply function $s_\Delta(w,\tilde{w},z,\tilde{z}) = \gamma^2\norm{w-\tilde{w}}^2-\norm{z-\tilde{z}}^2$ with $\gamma\in\mathbb{R}_+$ implies that $\Sigma$ has a bounded \litwo-gain $\leq \gamma$. As this supply function satisfies the conditions of Theorem \ref{thm:incrstab}, incremental asymptotic stability of $\Sigma$ is also implied.
%this also implies that $\Sigma$ is incrementally asymptotically stable.

\section{Conversion of $P$ to a coarse structure} \label{App:A}
While \eqref{eq:genplant} may seem restrictive, the general class of nonlinear plants \eqref{eq:genplantfull} can be expressed as \eqref{eq:genplant} by the use of appropriate filters.
	Consider  the following (low-pass) filters 
\begin{equation}\label{eq:lpfilt}
	F_i:\left\lbrace\begin{aligned}
		\dot{x}_{\mr{F},i}(t) &= -\Omega_i \,x_{\mr{F},i}(t)+\Omega_i\,u_{\mr{F},i}(t);\\
		y_{\mr{F},i}(t)&= x_{\mr{F},i}(t);
	\end{aligned}\right. 
\end{equation}
for $i=1,2$ and where $\Omega_i = \mr{diag}(\omega_{1,i},\dots,\omega_{n_{\mr{F,}i}})$, with $\omega_{i,j}>0$ for $j=1,\dots,n_{\mr{F,}i}$ and $x_{\mr{F},i}(t)\in\mathbb{R}^{n_{\mr{F,}i}}$. Connecting $F_1$ and $F_2$ \eqref{eq:lpfilt} to $P$ such that $\genInput=F_1(\hat{u}_\mr{\pnote},x_{\mr{F},1}(0))$, where $\hat{u}_\mr{\pnote}$ is the new control input signal, and $\hat{y}_\mr{\pnote}=F_2(\genOutput,x_{\mr{F},2}(0))$ results in
	\begin{equation}\label{eq:genplantfilt}
    \hat{P}: \left \lbrace
    \begin{aligned}
    \genStateDot(t) &= \genF \left(x_\mr{\pnote}(t),x_{\mr{F},1}(t)\right)+B_\mr{w} w(t);\\
    x_{\mr{F},1}(t) &= -\Omega_1 x_{\mr{F},1}(t)+\Omega_1 \hat{u}_\mr{\pnote}(t);\\
    x_{\mr{F},2}(t)&= \Omega_2\,\genHy(\genState(t),x_{\mr{F},1}(t))-\Omega_2\, x_{\mr{F},2}(t)\\&\phantom{=}+\Omega_2\,D_\mr{yw}w(t);\\
    z(t) &=  \genHz \left(\genState(t),x_{\mr{F},1}(t)\right)+D_\mr{zw}w(t);\\
    \hat{y}_\mr{\pnote} &= x_{\mr{F},2}(t);
    \end{aligned}\right. 
\end{equation}
which is of the form \eqref{eq:genplant}. Note that if $\omega_{i,j}$ is taken large enough (e.g., $5 \times$ the intended bandwidth) then the desired closed-loop performance is not affected by the conversion.

%%%%%%%%%%%%%%%%%%%%%%%%%%%%%%%%%%%
\section{Proof of Theorem \ref{thrm:incric}}  \label{App:B}
Without loss of generality we can omit $w$ and $z$ and assume $P$ is given by (dependence on $t$ is omitted for clarity):
\begin{equation}
P: \left \lbrace \begin{aligned}
\dot{x}_\mr{\pnote} &= f(x_\mr{\pnote},\genInput);\\
\genOutput &= h_\mr{\pnote,y}(x_\mr{\pnote},\genInput);
\end{aligned} \right. 
\end{equation}
and $K$ is given by \eqref{eq:generalcontroller}. $P$ and $K$ are interconnected such that $u_\mathrm{k} = \genOutput$ and $\genInput = y_\mathrm{k}$. We assume that the interconnection is well-posed, i.e., there exists a $\mathcal{C}_1$ function $\breve{h}$, such that $\genInput = h_\mr{k}(x_\mr{k},h_\mr{\pnote,y}(x_\mr{\pnote},\genInput))$ can be expressed as
		$\genInput = \breve{h}(x_\mr{\pnote},x_\mr{k})$.
The closed-loop is then given by
\begin{equation}\label{eq:cl}
\Gamma: \left \lbrace \begin{aligned}
\dot{x}_\mr{\pnote} &= f(x_\mr{\pnote},\breve{h}(x_\mr{\pnote},x_\mr{k}));\\
\dot{x}_\mathrm{k} &= f_\mathrm{k}(x_\mathrm{k},h_\mr{\pnote,y}(x_\mr{\pnote},\breve{h}(x_\mr{\pnote},x_\mr{k}))).
\end{aligned} \right. 
\end{equation}
The differential form of \eqref{eq:cl} is
\begin{equation}\label{eq:dCL}
\delta \Gamma\!:\!\left\lbrace
\begin{aligned}
\delta \dot{x}_\mr{\pnote}&=\frac{\partial f}{\partial x_\mr{\pnote}}(x_\mr{\pnote},\genInput)\delta x_\mr{\pnote} \!+\!\frac{\partial f}{\partial \genInput}(x_\mr{\pnote},\genInput)\\
&\phantom{=}\cdot\bigg(\!\frac{\partial \breve{h}}{\partial x_\mr{\pnote}}(x_\mr{\pnote},x_\mr{k}) \delta x_\mr{\pnote}+\!\frac{\partial \breve{h}}{\partial x_\mathrm{k}}(x_\mr{\pnote},x_\mathrm{k})\delta x_\mr{k}\!\bigg);\\
\delta \dot{x}_\mathrm{k} \!&=\! \frac{\partial f_\mr{k}}{\partial x_\mr{k}}(x_\mr{k},u_\mr{k})\delta x_\mr{k} \!+\! \frac{\partial f_\mr{k}}{\partial u_\mr{k}}(x_\mr{k},u_\mr{k})\\
&\phantom{=}\cdot\bigg[\frac{\partial h_\mr{\pnote,y}}{\partial x_\mr{\pnote}}(x_\mr{\pnote},\genInput) \delta x_\mr{\pnote}\!+\! \frac{\partial h_\mr{\pnote,y}}{\partial \genInput}(x_\mr{\pnote},\genInput)\\
&\phantom{=}\cdot\bigg(\!\frac{\partial \breve{h}}{\partial x}(x_\mr{\pnote},x_\mr{k})\delta x+\frac{\partial \breve{h}}{\partial x_\mr{k}}(x_\mr{\pnote},x_\mr{k})\delta x_\mr{k}\!\bigg)\bigg];
\end{aligned}
\right.
\end{equation}
where $u_\mr{k} = h_\mr{\pnote,y}(x_\mr{\pnote},\breve{h}(x_\mr{\pnote},x_\mr{k}))$. %and $\genInput = \breve{h}(x_\mr{\pnote},x_\mr{k})$.
The differential forms of $P$ and $K$ are given by
\begin{equation}\label{eq:deltaPthm}
\delta P: \!\left \lbrace \begin{aligned}
\delta \dot{x}_\mr{\pnote} \!&=\! \frac{\partial f}{\partial x_\mr{\pnote}}(x_\mr{\pnote},\genInput) \delta x_\mr{\pnote} \!+\!\frac{\partial f}{\partial \genInput}(x_\mr{\pnote},\genInput) \delta \genInput;\\
\delta \genOutput\!&=\! \frac{\partial h_\mr{\pnote,y}}{\partial x_\mr{\pnote}}(x_\mr{\pnote},\genInput) \delta x_\mr{\pnote}\!+\!\frac{\partial h_\mr{\pnote,y}}{\partial \genInput}(x_\mr{\pnote},\genInput) \delta \genInput;\end{aligned}\right. 
\end{equation}
\begin{equation}\label{eq:deltaKthm}
\delta K: \left \lbrace \begin{aligned}
\delta \dot{x}_\mathrm{k} &= \frac{\partial f_\mathrm{k}}{\partial x_\mathrm{k}}(x_\mathrm{k},u_\mathrm{k}) \delta x_\mathrm{k} +\frac{\partial f_\mathrm{k}}{\partial u_\mathrm{k}}(x_\mathrm{k},u_\mathrm{k}) \delta u_\mathrm{k}; \\
\delta y_\mathrm{k} &= \frac{\partial h_\mathrm{k}}{\partial x_\mathrm{k}}(x_\mathrm{k},u_\mathrm{k}) \delta x_\mathrm{k} +\frac{\partial h_\mathrm{k}}{\partial u_\mathrm{k}}(x_\mathrm{k},u_\mathrm{k}) \delta u_\mathrm{k};
\end{aligned}\right. 
\end{equation}
interconnecting these in a similar manner as $P$ and $K$, i.e. $\delta u_\mr{k} = \delta \genOutput$ and $\delta \genInput = \delta y_\mr{k}$, results in
\begin{equation}\label{eq:deltaU2}
	\begin{aligned}
	\delta \genInput =\,& 
	\frac{\partial h_\mathrm{k}}{\partial x_\mathrm{k}}(x_\mathrm{k},u_\mathrm{k}) \delta x_\mathrm{k} +\frac{\partial h_\mathrm{k}}{\partial u_\mathrm{k}}(x_\mathrm{k},u_\mathrm{k})\\& \cdot\bigg(\frac{\partial h_\mr{\pnote,y}}{\partial x_\mr{\pnote}}(x_\mr{\pnote},\genInput) \delta x_\mr{\pnote}+\frac{\partial h_\mr{\pnote,y}}{\partial \genInput}(x_\mr{\pnote},\genInput) \delta \genInput\bigg).
	\end{aligned}
\end{equation}
By the well-posedness assumption, we know that $\genInput = h_\mr{k}(x_\mr{k},h_\mr{\pnote,y}(x_\mr{\pnote},\genInput))$ can be expressed as $\genInput = \breve{h}(x_\mr{\pnote},x_\mr{k})$, hence, in the differential form, \eqref{eq:deltaU2} can equivalently be expressed as $\delta \genInput = \frac{\partial\breve{h}}{\partial x_\mr{\pnote}}(x_\mr{\pnote},x_\mr{k})\delta x_\mr{\pnote} +\frac{\partial\breve{h}}{\partial x_\mr{k}}(x_\mr{\pnote},x_\mr{k})\delta x_\mr{k}$.  Combining this with \eqref{eq:deltaPthm} and \eqref{eq:deltaKthm} allows us to express the interconnection of $\delta P$ and $\delta K$ as $\delta \Gamma$ \eqref{eq:dCL}. Note that by writing  \eqref{eq:deltaU2} as \vspace{-1mm}
\begin{equation}\label{eq:deltaU}
\begin{aligned}
	\delta \genInput = \,& \overbrace{h_\delta(x_\mr{\pnote},x_\mr{k})\frac{\partial h_\mathrm{k}}{\partial u_\mathrm{k}}(x_\mathrm{k},u_\mathrm{k})\frac{\partial h_{\mr{\pnote},\mr{y}}}{\partial x_\mr{\pnote}}(x_\mr{\pnote},\genInput)}^{ \frac{\partial\breve{h}}{\partial x_\mr{\pnote}}(x_\mr{\pnote},x_\mr{k})} \delta x_\mr{\pnote}
	\\&+\underbrace{h_\delta(x_\mr{\pnote},x_\mr{k})\frac{\partial h_\mathrm{k}}{\partial x_\mathrm{k}}(x_\mathrm{k},u_\mathrm{k})}_{\frac{\partial\breve{h}}{\partial x_\mr{k}}(x_\mr{\pnote},x_\mr{k})} \delta x_\mathrm{k},
\end{aligned}\vspace{-2mm}
\end{equation}
where\footnote{Note again that $u_\mr{k} = h_\mr{\pnote,y}(x_\mr{\pnote},\breve{h}(x_\mr{\pnote},x_\mr{k}))$ and $\genInput = \breve{h}(x_\mr{\pnote},x_\mr{k})$.}  $h_\delta(x_\mr{\pnote},x_\mr{k}) = \left(I-\frac{\partial h_\mathrm{k}}{\partial u_\mathrm{k}}(x_\mathrm{k},u_\mathrm{k})\frac{\partial h_\mr{\pnote,y}}{\partial \genInput}(x_\mr{\pnote},\genInput)\right)^{-1}$, we get constructive conditions on the existence of $\breve{h}$.

%%%%%%%%%%%%%%%%%%%%%%%%%%%%%%%%%%%

\section{Proof of Theorem \ref{thrm:diffICL2}} \label{App:CB}
By synthesis, we obtain a controller $\delta K$ such that the closed-loop interconnection $\mathcal{F}_\mr{l}(\delta P_\mr{LPV},\delta K)$ is \ltwo-gain stable with a bounded \ltwo-gain of $\gamma$ for all $\sche\in\mathcal{P}^\mathbb{R}$. As $\mathfrak{B}_\mr{\delta p}^\m{X}\subseteq\mathfrak{B}_\mr{LPV}$, see Definition \ref{def:lpvemb}, 
this implies that $\mathcal{F}_\mr{l}(\delta P,\delta K)$ with $\sche=\psi(x_\mr{\pnote})$ for $\delta K$ is \ltwo-gain stable with a  \ltwo-gain %less than or equal to 
$\leq\gamma$ for all $x_\mr{\pnote}\in\mathfrak{B}_{\mr{\pnote},\mr{x_p}}^\m{X}$.

%%%%%%%%%%%%%%%%%%%%%%%%%%%%%%%%%%%
\section{Proof of Theorem \ref{thm:contrrealize}} \label{App:C}
Based on the definition of the differential variables we have that $\delta x_\mr{k}(t) = \frac{\partial}{\partial \lambda} \bar{x}_\mr{k}(t,\lambda)$, $\delta u_\mr{k}(t,\lambda) = \frac{\partial}{\partial \lambda} \bar{u}_\mr{k}(t,\lambda)$, $\delta y_\mr{k}(t,\lambda) = \frac{\partial}{\partial \lambda} \bar{y}_\mr{k}(t,\lambda)$. The family of parameterized trajectories is defined as $(\bar{x}_\mr{k}(\lambda),\bar{u}_\mr{k}(\lambda),\bar{y}_\mr{k}(\lambda))$ with $\lambda\in[0,1]$ such that $(\bar{x}_\mr{k}(1),\bar{u}_\mr{k}(1),\bar{y}_\mr{k}(1))=({x}_\mr{k},{u}_\mr{k},{y}_\mr{k})$ is the current trajectory and $(\bar{x}_\mr{k}(0),\bar{u}_\mr{k}(0),\bar{y}_\mr{k}(0))=(\st{x_\mr{k}},\st{u_\mr{k}},\st{y_\mr{k}})$ is the steady-state trajectory. Consequently,
	\begin{subequations}
	\begin{align}
		y_\mr{k}(t) &= \st{y_\mr{k}}(t)+\int_0^1 \frac{\partial}{\partial \lambda}\bar y_\mr{k}(t,\lambda)\,d\lambda,\\
		&= \st{y_\mr{k}}(t)+\int_0^1 \delta y_\mr{k}(t,\lambda)\,d\lambda.
	\end{align}
	\end{subequations}
	Based on $\delta K$, in terms of \eqref{eq:incrContr}, we get
	\begin{equation}\label{eq:CDSy}
	\begin{aligned}
		y_\mr{k}(t) = \st{y_\mr{k}}(t)+\int_0^1 &C_\mathrm{k}(\psi(\bar{x}_\mr{\pnote}(t,\lambda)) \delta x_\mathrm{k}(t,\lambda) +\\ &D_\mathrm{k}(\psi(\bar{x}_\mr{\pnote}(t,\lambda)) \delta u_\mathrm{k}(t,\lambda)\,d\lambda.
	\end{aligned}
	\end{equation}
	The closed-loop differential storage function is $V(\fullstate,\delta \fullstate) = \delta \fullstate ^\top M \delta \fullstate$ with $M\succ 0$, corresponding to a constant Riemannian metric. Hence, the homotopy path connecting $\fullstate(t)$ and $ \st{\fullstate}(t)$ is given by a straight line, i.e., by 
	\begin{equation}
		\bar\fullstate(t,\lambda) = \st{\fullstate}(t)+\lambda(\fullstate(t)-\st{\fullstate}(t)),
	\end{equation}
	see \cite{Manchester2018}. Therefore, $\bar{x}_\mr{\pnote}(t,\lambda) = \st{x_\mr{\pnote}}(t)+\lambda(x_\mr{\pnote}(t)-\st{x_\mr{\pnote}}(t))$ and $\bar{x}_\mr{k}(t,\lambda) = \st{x_\mr{k}}(t)+\lambda(x_\mr{k}(t)-\st{x_\mr{k}}(t))$. This implies that $\delta x_\mr{k}(t,\lambda) = \frac{\partial}{\partial \lambda}\bar{x}_\mr{k}(t,\lambda)=x_\mr{k}(t)-\st{x_\mr{k}}(t)=\Delta x_\mr{k}(t)$ and similarly $\delta x_\mr{\pnote}(t,\lambda) = x_\mr{\pnote}(t)-\st{x_\mr{\pnote}}(t)$. Furthermore, define the parameterized trajectory $\bar w(t,\lambda)=\st{w}(t)+\lambda(w(t)-\st{w}(t))$, such that $\delta w(t,\lambda) := \frac{\partial \bar{w}}{\partial \lambda}(t,\lambda)=w(t)-\st{w}(t)$. Hence, as $\delta u_k(t,\lambda) = \delta \genOutput(t,\lambda) = C_\mr{y}\delta x_\mr{\pnote}(t,\lambda)+D_\mr{yw}\delta w(t,\lambda)$ is linear in $\delta x$ and $\delta w$, and as $\delta w(t,\lambda)= w(t)-\st{w}(t)$ and $\delta x_\mr{\pnote}(t,\lambda) = x_\mr{\pnote}(t)-\st{x_\mr{\pnote}}(t)$ we obtain $\delta u_k(t,\lambda)=u_\mr{k}(t)-\st{u_\mr{k}}(t)=\Delta u_\mr{k}(t)$. For the sake of readability, also introduce $\bar\sche(t,\lambda)=\psi(\bar{x}_\mr{\pnote}(t,\lambda))$. Using these relations for the differential state equation of $\delta K$ in \eqref{eq:incrContr} and filling these relations in \eqref{eq:CDSy}, result in
	\begin{subequations}\label{eq:xkyk}
		\begin{align}\label{eq:xk}
		\Delta \dot x_\mr{k}(t) &= \left(\int_0^1 A_\mathrm{k}(\bar\sche(t,\lambda))\,d\lambda\right)\Delta x_\mr{k}(t) +\notag\\ &\phantom{,=}\left(\int_0^1 B_\mathrm{k}(\bar\sche(t,\lambda))\,d\lambda\right)\Delta u_\mr{k}(t);\\
		\label{eq:yk}
		y_\mr{k}(t) &= \st{y_\mr{k}}(t)+\left(\int_0^1 C_\mathrm{k}(\bar\sche(t,\lambda))\,d\lambda\right)\Delta x_\mr{k}(t) +\notag\\ &\phantom{,=}\left(\int_0^1 D_\mathrm{k}(\bar\sche(t,\lambda))\,d\lambda\right)\Delta u_\mr{k}(t);&&
	\end{align}
	\end{subequations}
	giving us $K$ \eqref{eq:Contr}. Next, it is shown that the differential form of $K$ is $\delta K$. Based on \eqref{eq:xkyk} define:
	\begin{subequations}\label{eq:controlparameterized}
	\begin{align}
		\dot{\bar x}_\mr{k}(t,\lambda) &= \st{\dot{x}}_\mr{k}(t) +\left(\int_0^\lambda A_\mathrm{k}(\bar\sche(t,\lambda))\,d\lambda\right)\Delta x_\mr{k}(t) +\notag \\ &\phantom{=}\left(\int_0^\lambda B_\mathrm{k}(\bar\sche(t,\lambda))\,d\lambda\right)\Delta u_\mr{k}(t);\\
		\bar{y}_\mr{k}(t,\lambda) &= \st{y_\mr{k}}(t)+\left(\int_0^\lambda C_\mathrm{k}(\bar\sche(t,\lambda))\,d\lambda\right)\Delta x_\mr{k}(t) +\notag\\ &\phantom{=}\left(\int_0^\lambda D_\mathrm{k}(\bar\sche(t,\lambda))\,d\lambda\right)\Delta u_\mr{k}(t);&&
	\end{align}
	\end{subequations}
	Differentiating \eqref{eq:controlparameterized} w.r.t. $\lambda$ we obtain
	\begin{subequations}
	\begin{align}
		\frac{\partial \dot{\bar x}_\mr{k}}{\partial \lambda}(t,\lambda) \!=\! A_\mathrm{k}(\bar\sche(t,\lambda))\Delta x_\mr{k}(t)\!+\!B_\mathrm{k}(\bar\sche(t,\lambda))\Delta u_\mr{k}(t);\\
		\frac{\partial \bar y_\mr{k}}{\partial \lambda}(t,\lambda) \!=\! C_\mathrm{k}(\bar\sche(t,\lambda))\Delta x_\mr{k}(t)\!+\!D_\mathrm{k}(\bar\sche(t,\lambda))\Delta u_\mr{k}(t).
	\end{align}
	\end{subequations}
	Then, using that $\Delta u_\mr{k}(t)=\delta u_\mr{k}(t,\lambda)$, $\Delta x_\mr{k}(t)=\delta x_\mr{k}(t,\lambda)$, $\delta x_\mr{k}(t,\lambda) = \frac{\partial}{\partial \lambda} \bar{x}_\mr{k}(t,\lambda)$ and $\delta y_\mr{k}(t,\lambda) = \frac{\partial}{\partial \lambda} \bar{y}_\mr{k}(t,\lambda)$ and taking $\lambda=1$ we get
	\begin{subequations}
	\begin{align}
		\delta\dot{x}(t) &= A_\mathrm{k}(\sche(t))\delta x_\mr{k}(t)+B_\mathrm{k}(\sche(t))\delta u_\mr{k}(t);\\
		\delta{y}(t)& = C_\mathrm{k}(\sche(t))\delta x_\mr{k}(t)+D_\mathrm{k}(\sche(t))\delta u_\mr{k}(t);&&
	\end{align}
	\end{subequations}
	which is $\delta K$ \eqref{eq:incrContr}, completing the proof.

%%%%%%%%%%%%%%%%%%%%%%%%%%%%%%%%%%%
\section{Proof of Theorem \ref{thrm:cl-li2}} \label{App:D}
By Theorem \ref{thrm:diffICL2}, it holds that $\delta K$ ensures \ltwo-gain stability with a bounded \ltwo-gain $\gamma$ for $\mathcal{F}_\mr{l}(\delta P,\delta K)$ on $\mathcal{X}_\mr{\pnote}$. Furthermore, by Theorem \ref{thm:contrrealize}, the differential form of $K$ \eqref{eq:Contr} is given by $\delta K$ \eqref{eq:incrContr}. Consequently, by Theorem \ref{thrm:incric}, the differential form of $\mathcal{F}_\mr{l}(P,K)$ is given by $\mathcal{F}_\mr{l}(\delta P,\delta K)$. Moreover, we consider the set $\mathcal{W}\subseteq \mathbb{W}$, for which $\mathcal{X}=\mathcal{X}_\mr{\pnote} \times \mathcal{X}_\mr{k} $ is invariant, meaning  that for any  $w,\st{w}\in\mc{W}^{\mathbb{R}_+}$, the resulting $x_\mr{\pnote}(t),\st{x}_\mr{\pnote}(t)\in\mc{X}_\mr{\pnote},\ \forall\,t\in \mathbb{R}_+$. Hence, we will remain in the design set on which differential \ltwo-gain stability is ensured. 
Based on Theorem \ref{thm:diffincrdiss}, this then implies that there exists a function $\zeta(\fullstate,\tilde{\fullstate})\geq  0$ with $\zeta(0,0) = 0$ such that
\begin{equation} \label{eq:impr}
\begin{aligned}
\Vert\mathcal{F}_\mathrm{l}( P,  K)(w,\fullstate_0)-&\mathcal{F}_{\mathrm{l}}( P,  K)(\st{w},\st{\fullstate}_0)\Vert_{2} \\ &\leq \gamma \norm{w-\st{w}}_{2}+\zeta(\fullstate_0,\st{\fullstate}_0),
\end{aligned}
\end{equation}
for any $w,\st{w}\in\mc{W}^{\mathbb{R}_+}$ with $w-\st{w}\in\mathscr{L}_2^{n_\mr{w}}$ and $x_0, \st{x}_0\in\mc{X}$.
As  \eqref{eq:impr} implies \eqref{eq:closedincrgain}, $\mathcal{F}_\mr{l}(P, K)$ is \litwo-gain stable with a bounded \litwo-gain of $\gamma$ on $\mc{X}$.

As $\mathcal{F}_\mr{l}(P, K)$ is \litwo-gain stable on $\mc{X}$, it is also incrementally stable on $\mc{X}$ based on Theorem \ref{thm:incrstab}. The differential storage function is given by $V(\fullstate,\delta\fullstate)=\delta\fullstate^\top M\delta\fullstate$, which implies that the incremental storage function is given by $V(\fullstate,\tilde{\fullstate})={(\fullstate-\tilde{\fullstate})^\top} M(\fullstate-\tilde{\fullstate})$, see \cite{Verhoek2020}. The latter also qualifies as an incremental Lyapunov function, see Theorem \ref{thm:incrstab}.

Moreover, the (desired) steady-state trajectory $\stTraj\in\mathfrak{B}_\mr{\pnote}$ is a valid solution of $P$ with corresponding $(\st{\fullstate}, \st{w},\st{z})\in\mc{X}^{\mathbb{R}_+}\times\mc{W}^{\mathbb{R}_+}\times\mathbb{Z}^{\mathbb{R}_+}$ due to the well-posedness of $\mathcal{F}_\mathrm{l}(P, K)$. Consequently, this implies by Lemma \ref{lem:convergence} that all solutions converge towards $(\st{\fullstate}, \st{w},\st{z})$. Meaning, for all $w\in \mc{W}^{\mathbb{R}_+}$, when $w(t)\rightarrow \st{w}(t)$ as $t\rightarrow\infty$, $(\fullstate(t), w(t),z(t))\rightarrow(\st{\fullstate}(t),\st{w}(t),\st{z}(t))$ as $t\rightarrow\infty$.

\section{Proof of Theorem \ref{thm:nlobsv}}\label{App:E}
Define $F_\mr{e}(x_\mr{e},\hat{x}_\mr{e},\genInput,w_\mr{m}) = f_\mr{e}(\hat{x}_\mr{e},\genInput,w_\mr{m})+L(h_\mr{e}(x_\mr{e},w_\mr{m})-h_\mr{e}(\hat{x}_\mr{e},w_\mr{m}))$ and $H_\mr{e}(x_\mr{e},\hat{x}_\mr{e},w_\mr{m}) = h_\mr{e}(\hat{x}_\mr{e},w_\mr{m})$. As $F_\mr{e}(x_\mr{e},x_\mr{e},\genInput,w_\mr{m}) = f(x_\mr{e},\genInput,w_\mr{m})$ and $H_\mr{e}(x_\mr{e},x_\mr{e}, w_\mr{m})=h_\mr{e}(x_\mr{e},w_\mr{m})$, we have that \eqref{eq:observer} is a virtual system of \eqref{eq:distmodel}, see also \cite{Wang2005,Jouffroy2010}. The virtual system \eqref{eq:observer} is virtually contractive, meaning that $\hat{x}_\mr{e}(t)\rightarrow x_\mr{e}(t)$ for $t\rightarrow\infty$, see \cite{Wang2005,ReyesBaez2019}, if 
\begin{multline}\label{eq:virtualsystem}
		\delta \dot{\hat x}_\mr{e}(t)=\left(\frac{\partial f_\mr{e}}{\partial \hat{x}_\mr{e}}(\hat{x}_\mr{e}(t),\genInput(t),w_\mr{m}(t))\right)\delta\hat{x}_\mr{e}-\\L\left(\frac{\partial h_\mr{e}}{\partial \hat{x}_\mr{e}}(\hat{x}_\mr{e}(t),w_\mr{m}(t))\right)\delta \hat{x}_\mr{e}(t),
\end{multline}
is asymptotically stable. The differential form of the virtual system \eqref{eq:virtualsystem} can be  written as
\begin{equation}\label{eq:virtualdiffobsv}
		\delta \dot{\hat x}_\mr{e}(t)=\big(\m A_\mr{e}(\hat{x}_\mr{e}(t),\genInput(t),w_\mr{m}(t))-L\m C_\mr{e}(\hat{x}_\mr{e}(t),w_\mr{m}(t))\big)\delta \hat{x}_\mr{e}(t).
\end{equation} 
The system \eqref{eq:virtualdiffobsv} is asymptotically stable with (differential) Lyapunov function $V_\delta(\delta x)=\delta x^\top P\delta x$ if \eqref{eq:obsvLMI} holds for all $(x_\mr{e},\genInput,w_\mr{m})\in\mathbb{X}_\mr{e}\times\genInputSet\times\pi_\mr{w_\mr{m}}\mathbb{W}$. 

%%%%%%%%%%%%%%%%%%%%%%%%%%%%%%%%%%%
\bibliographystyle{plain}
\bibliography{References}

%%%%%%%%%%%%%%%%%%%%%%%%%%%%%%%%%%%%%%%%%%%%%%%%%%%%%%%%%%%%%%%%%%%%%%%%%%%%%%%%%%%%%%%%%%%%%%%%%
%%%%% Biographies %%%%%
\begin{IEEEbiography}[{\includegraphics[width=1in,height=1.25in,clip,keepaspectratio]{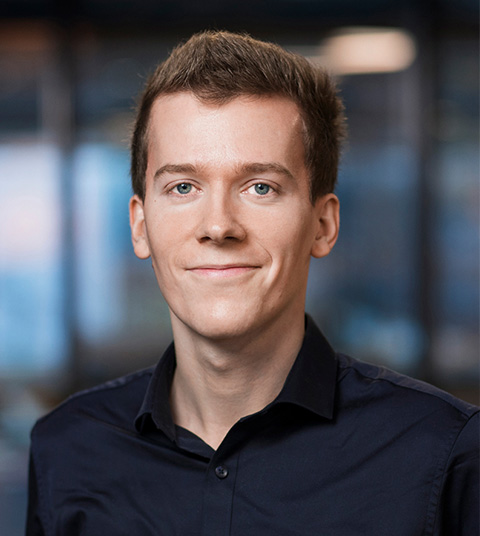}}]{Patrick J. W. Koelewijn} received his Bachelor's degree in Automotive and Master's degree in Systems and Control from the Eindhoven University of Technology, both Cum Laude, in 2016 and 2018 respectively. During his Master's degree he spent three months at the Institute of Control Systems at the Hamburg University of Technology (TUHH). He is currently pursuing a Ph.D. degree at the Control Systems Group, Department of Electrical Engineering, Eindhoven University of Technology. His main research interests include analysis and control of nonlinear and LPV systems, optimal and nonlinear control, and machine learning techniques.
\end{IEEEbiography}

\begin{IEEEbiography}[{\includegraphics[width=1in,height=1.25in,clip,keepaspectratio]{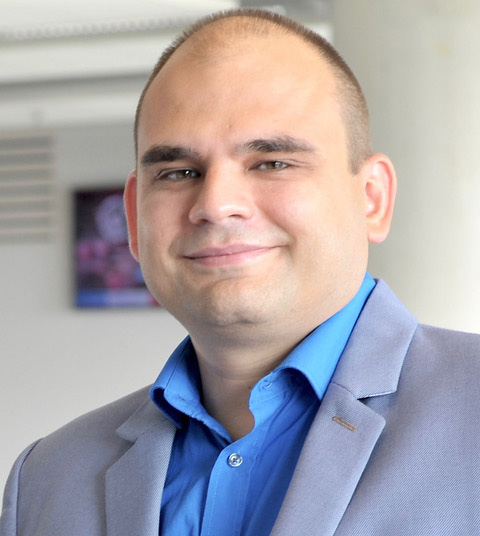}}]{Roland T\'oth} received his Ph.D. degree with cum laude distinction at the Delft Center for Systems and Control (DCSC), Delft University of Technology (TUDelft), Delft, The Netherlands in 2008.  He was a Post-Doctoral Research Fellow at TUDelft in 2009 and Berkeley in 2010. He held a position at DCSC, TUDelft in 2011-12. Currently, he is an Associate Professor at the Control Systems Group, Eindhoven University of Technology and a Senior Researcher at SZTAKI, Budapest, Hungary. His research interests are in identification and control of linear parameter-varying (LPV) and nonlinear systems, developing machine learning methods with performance and stability guarantees for modelling and control, model predictive control and behavioral system theory.
\end{IEEEbiography}

\begin{IEEEbiography}[{\includegraphics[width=1in,height=1.25in,clip,keepaspectratio]{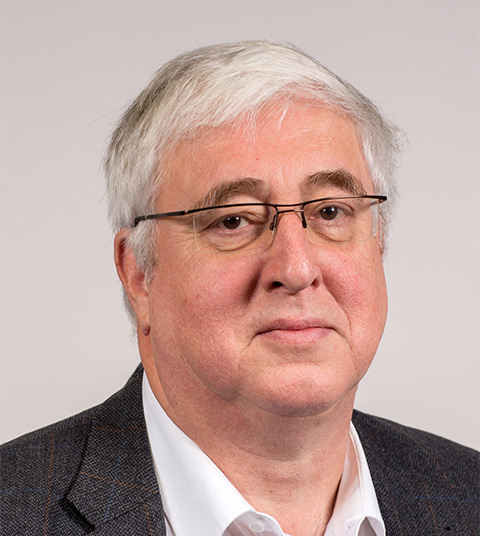}}]{Henk Nijmijer}  
Henk Nijmeijer (1955) is a full professor in Dynamics and Control at the Department of Mechanical Engineering of the Eindhoven University of Technology. His research field encompasses nonlinear dynamics and control and applications thereof. He is Field Chief Editor of Frontiers in Control Engineering. He is a fellow of the IEEE since 2000 and was awarded in 1990 the IEE Heaviside premium. He is appointed honorary knight of the `Golden Feedback Loop' (NTNU, Trondheim) in 2011. Since January 2015 he is scientific director of the Dutch Institute of Systems and Control (DISC). He is recipient of the 2015 IEEE Control Systems Technology Award and a member of the Mexican Academy of Sciences. He has been Graduate Program director of the TU/e Automotive Systems program in the period 2016-2021. He is an IFAC Fellow since 2019 and as of January 2021 an IEEE Life Fellow.
\end{IEEEbiography}

\begin{IEEEbiography}[{\includegraphics[width=1in,height=1.25in,clip,keepaspectratio]{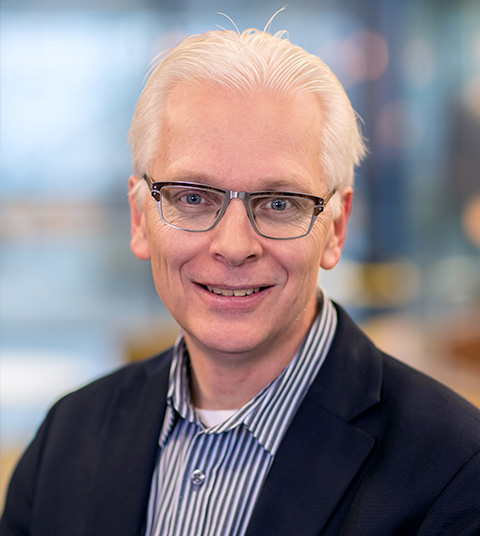}}]{Siep Weiland} received the M.Sc. (1986) and Ph.D. (1991) degrees in mathematics from the University of Groningen, The Netherlands. He was a Postdoctoral Research Associate at the Department of Electrical Engineering and Computer Engineering, Rice University, Houston, USA, from 1991 to 1992. Since 1992, he has been affiliated with Eindhoven University of Technology, Eindhoven, The Netherlands. He is a Full Professor at the same university with the Control Systems Group, Department of Electrical Engineering. His research interests are the general theory of systems and control, robust control, model approximation, modeling and control of spatial-temporal systems, identification, and model predictive control. %Prof. Weiland is the recipient of the Best Paper Award in the IEEE Transactions on Semiconductor Manufacturing for the years 2012 and 2020 and the Process Control Paper prize for the year 2020 in the IFAC Journal of Process Control.
\end{IEEEbiography}

\end{document}